\documentclass{tlp}
\usepackage{aopmath}
\usepackage{graphicx} 

\submitted{13 Oct 2012}
\revised{22 Feb 2013}
\accepted{22 Aug 2013}

\newcommand{\db}{\mathit{D\!B}}
\newcommand{\pdb}{P_\mathit{D\!B}}

\newcommand{\mswz}{{\tt msw }}
\newcommand{\pcfg}{\bf PG}
\newcommand{\expl}{\mathit{expl}}
\newcommand{\eq}{\mathit{eq}}
\newcommand{\derive}{ \stackrel{*}{\Rightarrow} }

\newcommand{\ddefined}{\stackrel{{\rm def}}{=}}

\newcommand{\argmax}[1]{\mathop{\rm argmax}\limits_{#1}}

\begin{document}

\bibliographystyle{acmtrans}
\long\def\comment#1{}
 
\title{ Infinite probability computation by cyclic explanation graphs }

\author[T. Sato and P. Meyer]
{ Taisuke Sato \\
  Tokyo Institute of Technology \\
  E-mail: sato@mi.cs.titech.ac.jp
\and
  Philipp Meyer\\
  Technical University Munich \\
  E-mail: meyerphi@in.tum.de
}

\pagerange{\pageref{firstpage}--\pageref{lastpage}}
\volume{\textbf{10} (3):}
\jdate{March 2002}
\setcounter{page}{1}
\pubyear{2002}

\maketitle

\label{firstpage}

\begin{abstract}
Tabling in logic programming has  been used to eliminate redundant computation
and also to stop infinite loop.  In this paper\footnote{
This paper is  based on \cite{Sato12a} and extended with  a theorem for prefix
PCFGs, a detailed explanation for  tabling, the addition of PLCGs, experiments
with a  real corpus and  two nonlinguistic applications: plan  recognition and
probabilistic model checking.
} we investigate another possibility  of tabling, i.e.\ to compute an infinite
sum  of  probabilities  for  probabilistic  logic programs.   Using  PRISM,  a
logic-based  probabilistic  modeling language  with  a  tabling mechanism,  we
generalize  prefix  probability  computation  for probabilistic  context  free
grammars (PCFGs) to probabilistic logic programs.  Given a top-goal, we search
for all proofs  with tabling and obtain an  explanation graph which compresses
them and  may be cyclic.  We  then convert the  explanation graph to a  set of
linear probability equations and solve them by matrix operation.  The solution
gives us the probability of the top-goal, which, in nature, is an infinite sum
of  probabilities.  Our  general  approach to  prefix probability  computation
through tabling not  only allows to deal with  non-PCFGs such as probabilistic
left-corner grammars (PLCGs) but has applications such as plan recognition and
probabilistic model checking and makes  it possible to compute probability for
probabilistic models describing cyclic relations.
To appear in Theory and Practice of Logic Programming (TPLP).
\end{abstract}

\begin{keywords}
tabling, probability computation, prefix, probability equation
\end{keywords}

\section{Introduction}
\label{sec:intro}
Combining logic  and probability in  a logic programming language  provides us
with a  powerful modeling tool  for machine learning.  The  resulting language
allows  us to  build  complex  yet comprehensible  probabilistic  models in  a
declarative way.   PRISM~\cite{Sato97,Sato01g,Sato08a} is one  of the earliest
attempts to develop such a language.   It covers a large class of known models
including   Bayesian  networks   (BNs),  hidden   Markov  models   (HMMs)  and
probabilistic context  free grammars  (PCFGs) and computes  probabilities with
the same time complexity as their standard algorithms\footnote{
They are the  junction tree algorithm for BNs,  the forward-backward algorithm
for HMMs and the inside-outside algorithm for PCFGs.
}, as well as unknown models such as probabilistic context free graph grammars
\cite{Sato08c}.

The   efficiency  of  probability   computation  in   PRISM  is  attributed to
the use of tabling~\cite{Tamaki86,Warren92,Rocha05,Zhou08,Zhou10}\footnote{
Tabling is  also employed by other probabilistic  logic programming languages
such as ProbLog~\cite{Mantadelis10} and PITA~\cite{Riguzzi11}.
} that eliminates redundant computation.   Given a top-goal $G$, we search for
all proofs of $G$\footnote{
In  this paper,  we  mean  by a  proof  of a  goal  $G$  an SLD-refutation  of
$\Leftarrow$${G}$.
} while  tabling probabilistic goals and recording  their logical dependencies
as a set $\expl(G)$ of propositional formulas with a graphical structure which
we  call an  {\em explanation  graph\/} for  $G$~\cite{Sato01g}.   By applying
dynamic  programming  to $\expl(G)$  when it  is   acyclic and  partially
ordered we  can efficiently compute the  probability of $G$ in  time linear in
the size  of the graph.   The use of  tabling also gives us  another advantage
over non-tabled  computation; it stops  infinite loop by  detecting recurrence
patterns of goals.  Tabled logic programs thus can directly use left recursive
rules in CFGs without the need of converting them to right recursive ones.\\

In  this paper we  investigate another  possibility of  tabling that  has gone
unnoticed in  the non-probabilistic  setting; we apply  tabling to  compute an
infinite sum of probabilities that  typically appears in the context of prefix
probability   computation   for  PCFGs~\cite{Jelinek91,Stolcke95,Nederhof11a}.
PCFGs are  a probabilistic extension of  CFGs in which CFG  rules are assigned
probabilities and the  probability of a sentence is  computed as a sum-product
of   probabilities    assigned   to   the    rules   used   to    derive   the
sentence~\cite{Baker79,Manning99}.  A prefix $u$  is an initial substring of a
sentence.  The  probability of  the prefix  $u$ is a  sum of  probabilities of
infinitely  many sentences  of  the form  $uv$  for some  string $v$.   Prefix
probability is useful in  speech recognition as discussed in \cite{Jelinek91}.
We  generalize this prefix  probability computation  for PCFGs  to probability
computation on  {\em cyclic explanation graph\/}s generated  by PRISM programs
using tabled  search.  Since we can  use arbitrary programs,  our approach not
only  allows us  to  deal  with non-PCFGs  such  as probabilistic  left-corner
grammars  (PLCGs)  in  addition  to  PCFGs,  but  opens  a  way  to  practical
applications such  as planning and model  checking as will  be demonstrated in
Section~\ref{sec:plan_recog}     and     in     Section~\ref{sec:reachability}
respectively.

PRISM constructs an explanation graph for a top-goal $G$ by collecting clauses
used in a proof of $G$ while checking  if there is a loop, i.e.\ if there is a
proved goal  that calls itself as  one of its descendent  goals.  Loops easily
occur for example in programs for prefix of PCFGs and in ones for Markov chain
containing self  loops.  By default whenever  PRISM detects a  loop during the
construction of the  explanation graph, it fails with an  error message but by
setting   {\tt    error\_on\_cycle}   flag    to   {\tt   off}    using   {\tt
  set\_prism\_flag/2}, we can let PRISM skip loop checking and as a result can
obtain a  cyclic explanation graph. So constructing  cyclic explanation graphs
requires no extra cost in PRISM.

However,   while   computing   probability   from  such   cyclic   graphs   is
possible~\cite{Etessami09}, efficient computation  is difficult except for the
case of {\em linear cyclic explanation graph\/}s that can be turned into a set
of   linear  probability  equations   straightforwardly  solvable   by  matrix
operation.  So  the practical  issue is to  guarantee the linearity  of cyclic
explanation  graphs.   We specifically  examine  a  PRISM  program for  prefix
probability computation for PCFGs and  prove that the program always generates
linear  cyclic  explanation  graphs.   We  also  prove  that  the  probability
equations obtained from  the linear cyclic explanation graphs  are solvable by
matrix operation under some mild assumptions on PCFGs.

To empirically test  our approach, we conduct experiments  of computing prefix
probability for  a PCFG and also  for a PLCG  using a real corpus  of moderate
size.  To our  knowledge, prefix probability computation for  PLCGs is new and
has not been  attempted so far.  As applications,  we apply prefix probability
computation to  plan recognition  in which action  sequences are  derived from
plans using a PCFG.  Our task is  to infer, given an action sequence, the plan
underlying it.  Note that we do not require the action sequence to be complete
as             a              sentence             unlike             previous
approaches~\cite{Bobick98,Lymberopoulos07,Amft07,Geib11}  as  we  are able  to
deal  with  prefix  action sequences.   We  also  apply  our approach  to  the
reachability      probability     problem      in      probabilistic     model
checking~\cite{Hinton06,Gorlin12}.  This  class of problems  needs to describe
Markov chains and to compute  the reachability probability between two states.
The experiment suggests that our approach is reasonably fast.

In  what  follows,  we  first  review  probability  computation  in  PRISM  in
Section~\ref{sec:probcomp}.   In   Section~\ref{sec:prepcfg}  we  explain  how
prefix probability  is computed for PCFGs  in PRISM together  with some formal
proofs.   Then we  tackle the  problem of  prefix probability  computation for
PLCGs in Section~\ref{sec:preplc}.  We apply prefix probability computation to
plan  recognition  in  Section~\ref{sec:plan_recog}  and to  the  reachability
probability     problem     in     probabilistic     model     checking     in
Section~\ref{sec:reachability}.      Section~\ref{sec:relatedwork}    contains
related work  and Section~\ref{sec:conclusion}  is the conclusion.   We assume
the reader has a basic familiarity with PRISM~\cite{Sato01g,Sato08a}.

\section{Probability computation in PRISM}
\label{sec:probcomp}
We   review   probability  computation   in   PRISM  for   self-containedness.
PRISM\footnote{
{\tt http://sato-www.cs.titech.ac.jp/prism/}
} is a probabilistic extension  of Prolog with built-in predicates for machine
learning     tasks    such    as     parameter    learning     and    Bayesian
inference~\cite{Sato01g,Sato08a}.   Theoretically a PRISM  program $\db$  is a
union $R \cup F$ of a set of  definite clauses $R$ and and a set $F$ of ground
probabilistic  atoms of  the form  {\tt msw($id$,$v$)}  that  represent simple
probabilistic choices where $id$ and $v$ are ground terms\footnote{
We use lower case strings to represent ground terms, atoms, etc in this paper.
}. Using probabilities  assigned to {\tt msw} atoms,  $\db$ uniquely defines a
probability measure $\pdb(\cdot)$  over possible Herbrand interpretations from
which  the   probability  of  an  arbitrary  closed   formula  is  calculated.
Practically however  PRISM programs  are just Prolog  programs that  use \mswz
atoms introduced by {\tt values/2} declarations\footnote{
A declaration {\tt values($id$,[$v_1,\ldots,v_N$])} introduces a set of ground
probabilistic atoms {\tt msw($id$,$v_i$)}($1\leq i\leq N$).  They represent as
a   group  a  discrete   random  variable   on  a   sample  space   $V_{id}  =
\{v_1,\ldots,v_N\}$.  So  only one of them becomes  probabilistically true and
others are false.   To specify their distribution we use  a PRISM command {\tt
  set\_sw($id$,[$\theta_1,\ldots,\theta_N$])}    that   sets   $\pdb(\mbox{\tt
  msw($id$,$v_i$)})$, the probability of  {\tt msw($id$,$v_i$)} being true, to
$\theta_i$ ($1\leq i\leq N$) where $\sum_{v \in V_{id}} \theta_v =1$.
} as probabilistic primitives\footnote{
Procedurally, executing  {\tt msw($id$,X)} as a PRISM  goal returns $\mbox{\tt
  X} = v_i$ with probability $\theta_i$.  On the other hand a ground goal {\tt
  msw($id$,$v$)}  is equivalent to  {\tt msw($id$,X),X=$v$}  and fails  if the
value  returned  in  {\tt X}  differs  from  $v$.   We assume  that  different
occurrences of {\tt msw/2} atom in a program or in a proof are independent and
if  they have  the  same $id$,  they  represent samples  from independent  and
identically distributed random variables~\cite{Sato01g}.
}        as        shown        in        Fig.~\ref{fig:prism:prog0}        of
Subsection~\ref{subsec:exampleprepcfg}.

In PRISM,  the probability $\pdb(G)$  of a ground  atom $G$ w.r.t.\  a program
$\db$  is  basically   computed  as  a  sum  of   probabilities  of  all  {\em
  explanation\/}s  for  $G$  where  an   {\em  explanation  for\/}  $G$  is  a
conjunction $E  = {\tt msw}_1  \wedge\cdots\wedge {\tt msw}_k$ of  ground {\tt
  msw}  atoms  such  that  ${\tt msw},\ldots,{\tt  msw}_k,  \mathit{comp}(R)
\vdash G$\footnote{
$\mathit{comp}(R)$  is  the  completion  of   $R$.   It  is  a  union  of  the
  if-and-only-if form of $R$ and the so called Clark's equational theory.
}. However naively computing $\pdb(G)$ is computationally expensive because of
exponentially many explanations.  Instead we compute $\pdb(G)$ in three steps.
In  the first  step, we  perform tabled  search for  all proofs  of  $G$ while
recording clause  instantiations used in a  proof in the  external memory area
(through some  C-interface predicates).  In  the second step, we  construct an
explanation graph  $\expl(G)$ for $G$ from recorded  clause instantiations. It
compactly represents all possible  explanations for $G$ by subformula sharing.
In the third step, we convert $\expl(G)$ to a set of probability equations and
obtain $\pdb(G)$ by solving it using dynamic programming.  In the following we
discuss each of them in detail.

\subsection{Tabled search and explanation graphs}
\label{subsec:tabling}

In general there are exponentially many proofs of $G$ and so are explanations.
Fortunately  we can  often compress  them to  an equivalent  but  much smaller
representation  by  factoring  out  common  sub-conjunctions  as  intermediate
goals~\cite{Sato01g,Zhou08}.  We can express the  set of all explanations as a
set  of  {\em defining  formula\/}s  that  take  the form  $H  \Leftrightarrow
\alpha_1  \vee\ldots\vee  \alpha_M$.  Here  $H$  is  the  top-goal $G$  or  an
intermediate  goal.   Hereafter  the   top-goal  and  intermediate  goals  are
collectively  called  {\em  defined  goal\/}s.   We call  each  $H  \Leftarrow
\alpha_i$  ($1  \leq i  \leq  M$)  a {\em  defining  clause\/}  for $H$  where
$\alpha_i$  is  a  conjunction  $C_1  \wedge\ldots\wedge  C_m  \wedge  \mswz_1
\wedge\ldots\wedge \mswz_n$ ($0 \leq m,n$) of defined goals $\{ C_1,\ldots,C_m
\}$ and \mswz atoms $\{ \mswz_1,\ldots,\mswz_n \}$.

We say that $H$ is a {\em parent\/}  of $C_j$ ($1 \leq j \leq m$) and call the
transitive closure of this  parent-child relation the {\em ancestor relation}.
The whole set of defining formulas is denoted by $\expl(G)$ and called an {\em
  explanation  graph\/} for  $G$  as is  called  so far.   In $\expl(G)$  each
defined goal  has only  one defining  formula and possibly  is referred  to by
other defined goals.

An  n-ary predicate  {\tt p/n}  is  said to  be {\em  probabilistic\/} if  the
predicate symbol  {\tt p} is  {\tt msw} or  recursively, there is a  clause in
$\db$  such  that  the head  contains  the  predicate  symbol  {\tt p}  and  a
probabilistic  predicate   occurs  in  the   body.   Likewise  an   atom  {\tt
  p($t_1,\ldots,t_n$)} is  probabilistic if {\tt p/n}  is probabilistic.  Then
roughly $\expl(G)$ is obtained from exhaustive tabled search for all proofs of
$G$ while  tabling probabilistic  predicates in $\db$.   What we  actually use
however  is not  $\db$  but another  non-probabilistic  Prolog program  $\db'$
translated from $\db$  that has a mechanism of  recording instantiated clauses
used in  a proof  of $G$.  We  construct $\expl(G)$  by tabled search  for all
proofs  of  $G$ w.r.t.\  $\db'$  while  tabling  probabilistic predicates  and
collect instantiated clauses used in  a proof as defining clauses constituting
$\expl(G)$~\cite{Kameya00,Zhou03}.

$\db'$   is    obtained   by   translating    each   clause   in    $\db$   as
follows~\cite{Zhou03}\footnote{
The actual implementation is  slightly different.  Also another translation is
possible which stores defining clauses in the table~\cite{Kameya00}.
}. Suppose for example
 {\tt p(X,f(V)):-}{\tt msw(X,V),q(g(X,V)),r(V)}
is  a  clause  in  $\db$  and  also  suppose  {\tt  p/2}  and  {\tt  q/1}  are
probabilistic but {\tt r/1} is  not (generalization is easy).  We replace {\tt
  msw(X,V)} with {\tt (get\_values(X,Vs),member(V,Vs))}\footnote{
For {\tt X} = $id$, {\tt get\_values(X,Vs)} returns the list of possible values
{\tt Vs} for {\tt msw($id$,$\cdot$)}.
} and further  add a special goal  to store a defining clause  in the external
memory area.  So the translated clause is\\

\begin{tabular}{l}
{\tt p(X,f(V)):-} {\tt  get\_values(X,Vs),member(V,Vs),q(g(X,V)),r(V),} \\
     \hspace{4em}{\tt  add\_to\_db(path(p(X,f(V)),[q(g(X,V))],[msw(X,V)]))}.\\[1em]
\end{tabular}
Here {\tt  member(V,Vs)} is a  backtrackable predicate and returns  an element
{\tt    V}   in    a    list    {\tt   Vs}.     The    combined   goal    {\tt
  (get\_values(X,Vs),member(V,Vs))} thus  succeeds with some value  {\tt V} in
the    outcome   space    {\tt   Vs}    for   {\tt    msw(X,$\cdot$)}.

When   all  goals  in   {\tt  (get\_values(X,Vs),member(V,Vs),q(g(X,V)),r(V))}
succeed, {\tt add\_to\_db/1} is invoked.  {\tt add\_to\_db(path($a$,$b$,$c$))}
is a special goal that always succeeds and stores a defining clause
{\tt $a$ <= $b$ \& $c$}
for  $a$ in the  external memory  area where  $b$ is  a list  (conjunction) of
probabilistic atoms and $c$ is a list (conjunction) of {\tt msw} atoms.

The  translated   program  $\db'$   is  a  usual   Prolog  program   and  runs
isomorphically to $\db$ as far as tabled search is concerned.  We mean by {\em
  tabled  goal\/}s goals  containing a  tabled predicate,  by  {\em answer\/}s
goals  successfully  proved  and   by  {\em  tabled  answer\/}s  tabled  goals
successfully proved respectively.  Then in tabled search if a call to a tabled
goal $H$ occurs, $H$ is unfolded by  a clause in the program and tabled search
continues, or  unified with a  tabled answer stored  in the table  and returns
with success.   In the former  case, if the  search succeeds and  $H\theta$ is
proved where $\theta$ is an answer substitution, the answer $H\theta$ is added
to the table.  In the latter case, the tabling strategy determines when tabled
answers     are     consumed.     More     details are    given     in
Subsection~\ref{subsec:linear-tabling}.
In the rest of the paper,  since $\db$ and $\db'$ behave identically, when the
context is clear,  we use $\db$ and $\db'$  interchangeably for simplicity and
say for example ``all proofs of $G$ w.r.t.\ $\db$'' instead of ``all proofs of
$G$ w.r.t.\ $\db'$''.

\subsection{From explanation graphs to  probability computation}
\label{subsec:expl}
The probability $\pdb(G)$ of a given goal $G$ is precisely defined in terms of
the distribution semantics of PRISM.  But the problem is that the semantics is
so  abstractly defined  that  we cannot  know  the actual  value of  $\pdb(G)$
easily.   Here  we describe  how  to compute  it  from  $\expl(G)$ under  some
assumptions.

To  compute $\pdb(G)$,  we convert  each defining  formula  $H \Leftrightarrow
\alpha_1 \vee\ldots\vee \alpha_M$ in $\expl(G)$  to {\em a set  of probability
  equations for\/} $H$:

\begin{eqnarray}
P(H)   & = & P(\alpha_1)+\cdots+P(\alpha_M) \label{eq:prob-eq-1} \\
       &   & \;\mbox{where} \nonumber \\
P(\alpha_i) & = & P(C_1)\cdots P(C_m)\pdb(\mswz_1)\cdots  \pdb(\mswz_n) \;\;(1 \leq i \leq M)\nonumber  \\
       &   & \;\;\;\mbox{for}\;
    \alpha_i = C_1 \wedge\ldots\wedge  C_m \wedge  \mswz_1 \wedge\ldots\wedge  \mswz_n.
                                           \label{eq:prob-eq-2} \nonumber
\end{eqnarray}
We denote by  $\eq(G)$ the entire set of  probability equations thus obtained.
Note   that  the   conversion  assumes   exclusiveness  among   disjuncts  $\{
\alpha_1,\ldots, \alpha_M \}$ and independence among conjuncts $\{ C_1,\ldots,
C_m,\mswz_1,\ldots,\mswz_n \}$\footnote{
In this paper we assume  these conditions are always satisfied.  In particular
we assume  the generative exclusiveness  condition stated later  which implies
the exclusiveness among disjuncts.
}. We  consider the $P(H)$'s  in $\eq(G)$ as numerical  variables representing
unknown probabilities  and refer to  them as {\em $P$-variable\/}s.   Then the
right  hand side  of  (\ref{eq:prob-eq-1})  is  a  multivariate  polynomial  in
$P$-variables   with   non-negative  coefficients   which   are  products   of
$\pdb(\mswz)$s.

We  say  that  $\expl(G)$ {\em  is  acyclic\/}  if  the ancestor  relation  in
$\expl(G)$  is  acyclic.  When  $\expl(G)$  is acyclic  as  is  the case  with
standard  generative models  such as  BNs, HMMs  and PCFGs,  defined  goals in
$\expl(G)$ are  hierarchically ordered by  the ancestor relation (with  $G$ as
top-most  element) and  the P-variables  in $\eq(G)$  are  also hierarchically
ordered.   As  a result  $\eq(G)$  is uniquely  and  efficiently  solved in  a
bottom-up manner  by dynamic programming using  the generalized inside-outside
(IO) algorithm~\cite{Sato01g}  in time linear in the size  of $\eq(G)$ and the
unique solution gives $P(G)=\pdb(G)$.\\

There are  however cases where  $\expl(G)$ is cyclic  and so is  $\eq(G)$, and
hence it is impossible to apply dynamic programming to $\eq(G)$, or even worse
$\eq(G)$  may  not  have a  unique  solution  when  $\eq(G)$  is a  system  of
polynomial  equations of  second  degree or  higher.   Nonetheless, no  matter
whether it is cyclic or not, we can prove at least the existence of a solution
for $\eq(G)$ thanks  to the special form and properties  of $\eq(G)$ under the
{\em generative exclusiveness condition}; at any choice point in any execution
path of  the top-goal, a choice  of alternative path  is made by the  value of
{\tt X} sampled from {\tt msw($id$,X)}.  We quickly remark that this condition
is naturally satisfied by PRISM  programs for generative models in general and
BNs, HMMs and  PCFGs in particular, because in a  generative model, an outcome
is  generated  by a  sequence  of probabilistic  choices  and  the process  is
simulated by {\tt msw} atoms.

The  generative exclusiveness condition  implies that  every disjunction  in a
defining formula is exclusive and  originated from a probabilistic choice made
by  some  {\tt  msw}.   So  a defining  formula  $H  \Leftrightarrow  \alpha_1
\vee\cdots\vee   \alpha_M$  is  written   as  $H   \Leftrightarrow  (\mbox{\tt
  msw($id_H$,$v_1$)}\wedge\beta_1)          \vee\cdots\vee          (\mbox{\tt
  msw($id_H$,$v_M$)}\wedge\beta_M)$  for some  {\tt  msw($id_H$,$\cdot$)} that
has a sample space $  V_{id_H}$ such that $V_{id_H}\supseteq \{ v_1,\ldots,v_M
\}$.
Denote the  vector of  P-variables in  $\expl(G)$ by ${\bf  X}^G$ and  write a
component  $P(H)$ as  $X_H$.  Then  the probability  equation about  $P(H)$ is
represented   as  $X_H   =  T_H({\bf   X}^G)  =   \sum_{i=1}^M  \pdb(\mbox{\tt
  msw($id_H$,$v_i$)})\varphi_i^H({\bf X}^G)$ where $\varphi_i^H({\bf X}^G)$ is
a  product  of some  $\pdb({\tt  msw})$s and  variables  in  ${\bf X}^G$.   We
represent $\eq(G)$ as ${\bf X}^G = T({\bf X}^G)$.
Now define a vector  sequence $\left\{ {\bf X}^G_k \right\}_{k=0}^{\infty}$ by
${\bf X}^G_0 = {\bf 0}$\footnote{
We use {\bf 0} (resp.  {\bf 1}) to  denote a vector of $0$s (resp. a vector of
$1$s).
} and ${\bf X}^G_{k+1} = T({\bf  X}^G_k)$ for $k \geq 1$.  Then ${\bf X}^G_{k}
= T^{(k)}({\bf 0})$ ($k \geq 1$).  First we prove two lemmas.

\begin{lemma}
\label{lemma-1}
$T(\cdot)$ is monotonic, i.e.\ ${\bf X}^G \leq {\bf Y}^G$
implies $T({\bf X}^G) \leq T({\bf Y}^G)$\footnote{
For  $N$  dimensional vectors  ${\bf  X} = (x_1,\ldots,x_N)$  and ${\bf  Y}  =
(y_1,\ldots,y_N)$, we  write ${\bf X}  \leq {\bf Y}$  (resp.\ ${\bf X}  < {\bf
Y}$) if $x_i  \leq y_i$ (resp.  $x_i <  y_i$) for every $i$ $(1  \leq i \leq N)$.
}.
\end{lemma}
\begin{proof}
It is enough to prove that ${\bf X}^G \leq {\bf Y}^G$ implies $ T_H({\bf X}^G)
\leq T_H({\bf  Y}^G)$ for an  arbitrary component $T_H({\bf X}^G)$  of $T({\bf
  X}^G)$.   Suppose ${\bf X}^G  \leq {\bf  Y}^G$ and  write $T_H({\bf  X}^G) =
\sum_{i=1}^M   \pdb(\mbox{\tt  msw($id_H$,$v_i$)})   \varphi_i^H({\bf  X}^G)$.
Since every $\varphi_i^H({\bf  X}^G)$ is  a product of  some $\pdb({\tt  msw})$s and
variables in ${\bf X}^G$, ${\bf X}^G \leq {\bf Y}^G$ implies $\varphi_i^H({\bf
  X}^G) \leq \varphi_i^H({\bf Y}^G)$ for every $i$. Hence
\begin{eqnarray*}
T_H({\bf X}^G)
   & =    & \sum_{i=1}^M \pdb(\mbox{\tt msw($id_H$,$v_i$)}) \varphi_i^H({\bf X}^G) \\
   & \leq & \sum_{i=1}^M \pdb(\mbox{\tt msw($id_H$,$v_i$)}) \varphi_i^H({\bf Y}^G)
            = T_H({\bf Y}^G).
\end{eqnarray*}
\end{proof}

\begin{lemma}
\label{lemma-2}
Suppose the generative exclusiveness condition is satisfied.
$\left\{  {\bf X}^G_k \right\} _{k=0}^{\infty}$ is bounded from above;
${\bf X}^G_k \leq {\bf 1}$ for every $k \geq 0$.
\end{lemma}
\begin{proof}
For $k=0$,  ${\bf X}^G_0  = {\bf 0}  \leq {\bf  1}$ holds.  Suppose  $k>0$ and
inductively  assume  ${\bf X}^G_k  \leq  {\bf  1}$  holds.  Let  $X^H_{k+1}  =
T_H({\bf X}^G_k)$ be a probability equation  in ${\bf X}^G = T({\bf X}^G)$. We
see
\begin{eqnarray*}
X^H_{K+1}
 & = & T_H({\bf X}^G_k) \hspace{0.5em}
  = \hspace{0.5em}
     \sum_{i=1}^M \pdb(\mbox{\tt msw($id_H$,$v_i$)})\varphi_i^H({\bf X}^G_k) \\
 & \leq & \sum_{i=1}^M \pdb(\mbox{\tt msw($id_H$,$v_i$)})
 \hspace{0.5em} \leq \hspace{0.5em}
     \sum_{v \in V_{id_H}} \pdb(\mbox{\tt msw($id_H$,$v$)}) 
                                    \hspace{0.5em}=\hspace{0.5em} 1
\end{eqnarray*}
Here we use the fact that  since $\varphi_i^H({\bf X}^G)$ is a product of some
$\pdb({\tt msw})$s  and variables in ${\bf  X}^G$, ${\bf X}^G_k  \leq {\bf 1}$
implies $\varphi_i^H({\bf X}^G_k) \leq 1$.
\end{proof}

\begin{theorem}
\label{th-0}
Under   the   generative  exclusiveness   condition,   $\left\{  {\bf   X}^G_k
\right\}_{k=0}^{\infty}$  monotonically  converges to  the  least fixed  point
${\bf  X}_{\infty}^G  = T({\bf  X}_{\infty}^G)$  which  gives  a solution  for
$\eq(G)$.
\end{theorem}
\begin{proof}

$\left\{  {\bf X}^G_k \right\}_{k=0}^{\infty}$  is a  monotonically increasing
sequence  ${\bf  0}  =  {\bf   X}^G_0  \leq  {\bf  X}^G_1  \leq  \cdots$  by
Lemma~\ref{lemma-1}  which  is bounded  from  above by  Lemma~\ref{lemma-2}.
Consequently  $\left\{ {\bf X}^G_k  \right\}_{k=0}^{\infty}$ converges  to a
limit ${\bf X}^G_{\infty}$.  Furthermore  because $T$ is continuous, we have
$T({\bf  X}^G_{\infty})  =  T(\lim_{k  \rightarrow \infty}  {\bf  X}^G_k)  =
\lim_{k  \rightarrow \infty}  T({\bf X}^G_k)  = \lim_{k  \rightarrow \infty}
{\bf X}^G_{k+1} = {\bf X}^G_{\infty}$.
So we have ${\bf X}^G_{\infty} = T({\bf X}^G_{\infty})$.
Let ${\bf X'}^G  \geq {\bf 0}$ be another fixed
point of  $T$.  ${\bf  X}^G_{k} \leq  {\bf X'}^G$ for  all $k\geq  0$ is
inductively proved.  Therefore ${\bf X}_{\infty}^G = \lim_{k \rightarrow
\infty} {\bf X}^G_k \leq {\bf  X'}^G$.  Hence ${\bf X}^G_{\infty} $ is
the least fixed point of $T$.
\end{proof}

\section{Prefix probability computation for PCFGs in PRISM}
\label{sec:prepcfg}
In this section, using a concrete example,  we have a close look at how cyclic
explanation  graphs are  constructed  and investigate  their properties.   The
reader is assumed to have a basic knowledge of CFG parsing.

\subsection{A prefix parser}
\label{subsec:exampleprepcfg}

Before proceeding we introduce some terminology about CFGs for later use.  Let
$X$ be a nonterminal in a CFG, $\alpha$, $\beta$ a mixed sequence of terminals
and  nonterminals.  A  rule  for  $X$ is  a  production rule  of  the form  $X
\rightarrow \alpha$.  If there is a  rule of the form $X \rightarrow Y \beta$,
we  say  $X$  and $Y$  are  in  the  direct  left-corner relation.   The  {\em
  transitive  closure\/} of  the direct  left-corner relation  is  called {\em
  left-corner relation\/} and we write $X \rightarrow_{L}Y$ if $X$ and $Y$ are
in the left-corner relation. The  left-corner relation is {\em cyclic\/} if $X
\rightarrow_{L}X$ holds for some nonterminal $X$.   We say that a rule is {\em
  useless\/} if it  does not occur in any  sentence derivation.  A nonterminal
is {\em  useless\/} if  every rule for  it is  useless.  Otherwise it  is {\em
  useful}.  In  this paper we assume that  CFGs have ``{\tt s}''  as a default
start symbol and have no epsilon rule and no useless nonterminal.

Finally      let      $X\rightarrow\alpha_1:\theta_1,\ldots,      X\rightarrow
\alpha_n:\theta_n $ be the set of rules  for $X$ in a PCFG with {\em selection
  probabilities\/}  $\theta_1,\ldots,\theta_n$ where $\sum_{i=1}^n  \theta_i =
1$.  We assume  that every rule has a  {\em positive\/} selection probability.
If the sum  of probabilities of sentences derived from the  start symbol is 1,
the PCFG is said to  be {\em consistent\/}~\cite{Wetherell80}.  We also assume
that PCFGs are consistent.

Now we look at a concrete example of prefix probability computation based
on cyclic explanation  graphs.  Consider a CFG, ${\bf G}_0$ =  \{ 
\mbox{ {\tt s} $\rightarrow$ {\tt  s}\,{\tt s} },
\mbox{ {\tt s} $\rightarrow$ {\tt a} },
\mbox{ {\tt s} $\rightarrow$ {\tt b} }  \}
and its PCFG version, $\pcfg_0$ =  \{ 
\mbox{ {\tt s} $\rightarrow$ {\tt  s}\,{\tt s} : 0.4},
\mbox{ {\tt s} $\rightarrow$ {\tt a} : 0.3},
\mbox{ {\tt s} $\rightarrow$ {\tt b} : 0.3}  \}.
Here ``{\tt s}'' is  a start symbol in ${\bf G}_0$ and  ``{\tt a}'' and ``{\tt
  b}'' are  terminals.  \mbox{ {\tt s}  $\rightarrow$ {\tt s}\,{\tt  s} : 0.4}
says that the rule \mbox{ {\tt s} $\rightarrow$ {\tt s}\,{\tt s} } is selected
with probability 0.4 when ``{\tt s}'' is expanded in a sentence derivation.

\begin{figure*}[t]
\rule{\textwidth}{0.25mm}\\ [-1em]
\begin{verbatim}
values(s,[[s,s],[a],[b]]).
:- set_sw(s,[0.4,0.3,0.3]).

pre_pcfg(L):- pre_pcfg([s],L,[]).          --(1) % L is a prefix
pre_pcfg([A|R],L0,L2):-                    --(2) % L0 is ground when called
   ( values(A,_)-> msw(A,RHS),             --(3) % if A is a nonterminal
       pre_pcfg(RHS,L0,L1)                 --(4) % select rule A->RHS
   ;  L0=[A|L1] ),                         --(5) % else consume A in L0
   ( L1=[] -> L2=[]                        --(6) % (pseudo) success
   ; pre_pcfg(R,L1,L2) ).                  --(7) % recursion
pre_pcfg([],L1,L1).                        --(8) % termination
\end{verbatim}
\rule{\textwidth}{0.25mm}\\ 
\caption{Prefix PCFG parser $\db_0$}
\label{fig:prism:prog0}
\end{figure*}

A PRISM program  $\db_0$ in Fig.~\ref{fig:prism:prog0} is a  prefix parser for
$\pcfg_0$.  It is a slight modification  of a standard top-down CFG parser and
parses prefixes  accepted by  ${\bf G}_0$  such as ``{\tt  a}'' (as  list {\tt
  [a]}).  The  only difference is that  it can have {\em  pseudo success\/} at
line {\tt  (6)}, i.e.\ it immediately  terminates with success as  soon as the
input prefix is consumed even when  there remain some nonterminals in {\tt R}
at line {\tt (2)}\footnote{
This is justifiable because as we  assume that every nonterminal is useful, we
can prove that every nonterminal derives a terminal string with probability 1.
}.

A {\tt  values/2} declaration  {\tt values(s,[[s,s],[a],[b]])} in  the program
introduces  three {\tt msw}  atoms: {\tt  msw(s,[s,s])}, {\tt  msw(s,[a])} and
{\tt  msw(s,[b])}.  The  next command  {\tt :-  set\_sw(s,[0.4,0.3,0.3])} sets
$\theta_{{\tt s}\rightarrow  {\tt s  s}}$ = $P_{\db_0}({\tt  msw(s,[s,s])})$ =
0.4, $\theta_{{\tt  s}\rightarrow {\tt a}}$ =  $P_{\db_0}({\tt msw(s,[a])})$ =
0.3 and $\theta_{{\tt s}\rightarrow  {\tt b}}$ = $P_{\db_0}({\tt msw(s,[b])})$
= 0.3 respectively when the program is loaded.  Thus $\pcfg_0$ is encoded.  We
point out  that $\db_0$ is general,  applicable to any PCFG  just by replacing
the {\tt values/2} declaration and {\tt set\_sw} command with appropriate ones
that encode a given PCFG.

\subsection{Tracing linear-tabling}
\label{subsec:linear-tabling}

Once a program $\db$ and a top-goal $G$ are given for which the probability is
computed,  the next  task is  to  construct an  explanation graph  for $G$  by
searching for  all proofs while  tabling answers and recording  their defining
clauses in  the external memory area.   Using a simple  example, we illustrate
how tabled search for all proofs  is done by {\em linear-tabling with the lazy
  strategy\/} in B-Prolog~\cite{Zhou08} which has been a standard platform for
PRISM.

One of the  unique features of linear-tabling is  to iterate exhaustive tabled
search to obtain all answers when there are looping subgoals\footnote{
In this subsection, the terms ``subgoal'' and ``goal'' are used synonymously.
}.  More precisely, if  a call {\tt :-(A,...)} on a path  of an SLD-tree has a
sub-path   containing   sub-derivation   $\mbox{\tt   :-A}\Rightarrow   \cdots
\Rightarrow  \mbox{\tt  :-(A',...)}$  such  that  {\tt A}  and  {\tt  A'}  are
variants,  {\tt  A}  and  {\tt  A'} are  called  {\em  interdependent  looping
  subgoal\/}s.  Interdependent looping  subgoals constitute  a cluster.
The first looping subgoal {\tt A}  in the cluster that appears in the SLD-tree
is said to be a {\em top-most looping subgoal\/}~\cite{Zhou08}.
%

Although  a looping  subgoal causes  an  infinite loop,  it can  be proved  by
non-looping paths in  the SLD tree. We preserve  answers from such non-looping
paths in  the table  and make  them available as  tabled answers  when looping
subgoals are called.   Linear-tabling with the lazy strategy  tries to collect
all answers  for looping subgoals by  iterating {\em round\/}s  for a top-most
looping subgoal.  In a round exhaustive search by backtracking is performed to
generate all  proofs of  the top-most looping  subgoal while  consuming tabled
answers and adding  newly found answers to the table.   The lazy strategy does
not allow other subgoals outside the looping path to consume tabled answers of
the top-most looping subgoal until no more round generates new answers for the
looping subgoals~\cite{Zhou08}.\\

\begin{figure*}[t]
\rule{\textwidth}{0.25mm}\\ [-1em]
\begin{verbatim}
:- pre_pcfg([a])
  :- pre_pcfg([s],[a],[])
    :- msw(s,RHS1),pre_pcfg(RHS1,[a],L1)..
      (first round)
      :- pre_pcfg([s,s],[a],L1)..    % RHS1=[s,s], top-most looping subgoal TG
        :- msw(s,RHS2),pre_pcfg(RHS2,[a],L1)..
          :- pre_pcfg([s,s],[a],L1)..  
             % RHS2=[s,s], fails at (4) as no anwser available in the table
             % for :- pre_pcfg([s,s],[a],L1) yet.
          :- pre_pcfg([a],[a],L1)..
             % RHS2=[a], executes (5) and succeeds at (6) with L1=[], resulting
             % in tabled answers pre_pcfg([a],[a],[]) and pre_pcfg([s,s],[a],[])
             % with defining clauses
             %   pre_pcfg([a],[a],[]) and
             %   pre_pcfg([s,s],[a],[]) <= pre_pcfg([a],[a],[]) & msw(s,[a])
          :- pre_pcfg([b],[a],L1)..
             % RHS2=[b], fails at (5)
      (second round)
      :- pre_pcfg([s,s],[a],L1)..    % RHS1=[s,s], top-most looping subgoal TG
        :- msw(s,RHS2),pre_pcfg(RHS2,[a],L1)..
          :- pre_pcfg([s,s],[a],L1)..
             % RHS2=[s,s], this time can consume the tabled answer 
             % pre_pcfg([s,s],[a],[]) in the previous round and
             % succeeds with L1=[], giving pseudo success at (6) and
             % a defining clause
             %   pre_pcfg([s,s],[a],[]) <= pre_pcfg([s,s],[a],[]) & msw(s,[s,s])
             % no further answer generated
          :- pre_pcfg([a],[a],L1)..  % RHS2=[a], succeeds with L1=[]
          :- pre_pcfg([b],[a],L1)..  % RHS2=[b], fails at (5)
      (third round)
      :- pre_pcfg([s,s],[a],L1)..
             % yields no new answer, so :- pre_pcfg([s,s],[a],L1)
             % is completely evaluated with one answer pre_pcfg([s,s],[a],[])
             % which results in the success of :- pre_pcfg([s],[a],[])
             % giving a defining clause
             %   pre_pcfg([s],[a],[]) <= pre_pcfg([s,s],[a],[]) & msw(s,[s,s])
      :- pre_pcfg([a],[a],L1)..
             % RHS1=[a], succeeds with L1=[], results in the success of
             % :- pre_pcfg([s],[a],[]) giving a defining clause
             %   pre_pcfg([s],[a],[]) <= pre_pcfg([a],[a],[]) & msw(s,[a])
      :- pre_pcfg([b],[a],L1)..      % RHS1=[b], fails at (5)
      ...
\end{verbatim}
\rule{\textwidth}{0.25mm}\\ 
\caption{A sketch of SLD tree(s) for {\tt :- pre\_pcfg([a])} }
\label{fig:prism:trace0}
\end{figure*}

Fig.~\ref{fig:prism:trace0}  sketches  tabled  search  for  all  proofs  of  a
top-goal  $G_0$ =  {\tt  pre\_pcfg([a])} w.r.t.\  $\db_0$  while tabling  {\tt
  pre\_pcfg/1}   and  {\tt   pre\_pcfg/3}.   Here   {\tt   (1)},{\tt  (2)},...
correspond to line numbers in Fig.~\ref{fig:prism:prog0}\footnote{
Recall that as we explained in Subsection~\ref{subsec:tabling}, the program we
actually use in  the tabled search is a translated program  $\db'_0$ but as it
behaves exactly the same way as  the original one except that defining clauses
are recorded  in the  external memory  area, we explain  the tabled  search in
terms of $\db_0$ for intuitiveness and conciseness.
}.  Although  Fig.~\ref{fig:prism:trace0}  is  self-explanatory, we  add  some
comments.  The top-call to $G_0$  = {\verb!pre_pcfg([a])!}  leads to a call to
a  subgoal   {\tt  TG}  =  {\verb!pre_pcfg([s,s],[a],L1)!}   via   a  call  to
{\verb!pre_pcfg([s],[a],[])!} in  which {\tt values(s,\_)} is  tested true and
{\tt msw(s,RHS1)} is executed at line {\tt (3)}.  Since {\tt TG} is a top-most
looping subgoal, exhaustive tabled search is iterated on {\tt TG} until no new
answer is obtained.

In the  first round, a  proof by a  branch in the  SLD tree specified  by {\tt
  RHS2} =  {\tt [a]}  succeeds with  {\tt L1} =  {\tt []}  and gives  a tabled
answer {\tt  pre\_pcfg([s,s],[a],[])} for which a defining  clause is recorded
in the external  memory area.  In the second round a  branch specified by {\tt
  RHS2} =  {\tt [s,s]}  succeeds as well  using the previously  tabled answer,
giving   a  new   defining  clause   {\tt  pre\_pcfg([s,s],[a],[])   <=}  {\tt
  pre\_pcfg([s,s],[a],[]) \& msw(s,[s,s])}.  The  third round generates no new
answer and the call to {\tt TG} terminates successfully.  {\tt TG} now exports
its tabled answer {\verb!pre_pcfg([s,s],[a],[])!}   which leads to the success
of the top-call.

After  all  proof  search  is  done, PRISM  constructs  an  explanation  graph
$\expl(G_0)$ by  tracing tabled answers  starting from $G_0$  while collecting
defining clauses recorded in the  external memory area.  When PRISM encounters
looping subgoals in the body of  a defining clause, it looks at the PRISM-flag
{\tt  error\_on\_cycle} and if  the value  is ``{\tt  off}'', these  goals are
treated as  succeeded normally and as  a result a cyclic  explanation graph is
obtained.

\subsection{Computing prefix probability: an example}
\label{subsec:exampleprefix}
In  this subsection,  using the  continuing example,  we  describe probability
computation in cyclic explanation graphs.

An  explanation  graph  for  $G_0$  = {\verb!pre_pcfg([a])!}  is  obtained  by
executing a command {\verb!?- probf(pre_pcfg([a]))!}\footnote{
{\tt probf/1} is a built-in predicate in PRISM and
{\tt probf($G$)} displays the explanation graph of $G$.
} w.r.t.\  $\db_0$.  The command initiates exhaustive  tabled search described
in Subsection~\ref{subsec:linear-tabling}  and generates an  explanation graph
shown  in  Fig.~\ref{fig:prism:probf0}   consisting  of  defining  clauses  in
Fig.~\ref{fig:prism:trace0}.

\begin{figure*}[h]
\rule{\textwidth}{0.25mm}\\ [-1em]
\begin{verbatim}
pre_pcfg([a]) <=> pre_pcfg([s],[a],[])
pre_pcfg([s],[a],[]) <=>
   pre_pcfg([s,s],[a],[]) & msw(s,[s,s]) v pre_pcfg([a],[a],[]) & msw(s,[a])
pre_pcfg([s,s],[a],[]) <=>
   pre_pcfg([a],[a],[]) & msw(s,[a]) v pre_pcfg([s,s],[a],[]) & msw(s,[s,s])
pre_pcfg([a],[a],[])
\end{verbatim}
\rule{\textwidth}{0.25mm}\\ 
\caption{Explanation graph for prefix ``{\tt a}''}
\label{fig:prism:probf0} %
\end{figure*}

As  can be seen,  the top-most  looping subgoal  {\tt pre\_pcfg([s,s],[a],[])}
calls    itself.    We    convert    the   cyclic    explanation   graph    in
Fig.~\ref{fig:prism:probf0} to the  corresponding set of probability equations
shown   in   Fig.~\ref{fig:prism:probeq0}.    Here  we   used   abbreviations:
$\theta_{{\tt s}  \rightarrow {\tt s s}} =  P_{\db_0}({\tt msw(s,[s,s])})$ and
$\theta_{{\tt s} \rightarrow {\tt a}} = P_{\db_0}({\tt msw(s,[a])})$.

\begin{figure*}[h]
\rule{\textwidth}{0.25mm}\\ [1em]
\hspace*{1em}
$
\begin{array}{rclcl}
       P({\tt pre\_pcfg([a])})   & = & {\tt X} &  = & {\tt Y} \\
P({\tt pre\_pcfg([s],[a],[])})   & = & {\tt Y} & = & 
   {\tt Z}\cdot \theta_{{\tt s}\rightarrow {\tt s s}}
      + {\tt W}\cdot \theta_{{\tt s}\rightarrow {\tt a}}  \\
P({\tt pre\_pcfg([s,s],[a],[])}) & = & {\tt Z} & = & 
   {\tt W}\cdot  \theta_{{\tt s}\rightarrow {\tt a}}
      + {\tt Z}\cdot \theta_{{\tt s}\rightarrow {\tt s s}} \\
P({\tt pre\_pcfg([a],[a],[])})   & = & {\tt W} & = & \mbox{1}
\end{array}
$ \\[1em]
\rule{\textwidth}{0.25mm}\\
\caption{Probability equations for prefix ``{\tt a}''}
\label{fig:prism:probeq0} %
\end{figure*}

Although   we    know   that   the    set   of   probability    equations   in
Fig.~\ref{fig:prism:probeq0}  has a solution  (see Theorem~\ref{th-0}),  we do
not know their actual values.  To know their actual values, we need to compute
them  by solving  the equations.   Fortunately,  equations are  linear in  the
P-variables {\tt  X}, {\tt Y},  {\tt Z} and  {\tt W} and easily  solvable.  By
substituting $\theta_{{\tt  s}\rightarrow {\tt s s}}$ =  0.4 and $\theta_{{\tt
    s}\rightarrow  {\tt a}}$  = 0.3  for the  equations and  solving  them, we
obtain ${\tt X} = {\tt Y} = {\tt Z}$ = 0.5, and ${\tt W}$ = 1\footnote{
${\tt W}$ = 1 because  {\tt pre\_pcfg([a],[a],[])} is logically proved without
  involving {\tt msw\/}s.
}  respectively.   Hence  the  prefix  probability  of  ``{\tt  a}'',  $P({\tt
  pre\_pcfg([a])})$, is 0.5.  Note that this prefix probability is greater than
the probability of ``{\tt a}'' as a sentence which is 0.3. This is because the
prefix probability  of ``{\tt a}'' is  the sum of the  probability of sentence
``{\tt  a}''  {\em  and\/}  the  probabilities of  infinitely  many  sentences
extending ``{\tt a}''.

By looking at the set of probability equations in Fig.~\ref{fig:prism:probeq0}
more  closely,  we  can  understand  the  way  our  approach  computes  prefix
probability   in  PCFGs.    For  example,   consider  ${\tt   Z}   =  P({\tt
  pre\_pcfg([s,s],[a],[])})$  and  the  equation   ${\tt  Z}  =  {\tt  W}\cdot
\theta_{{\tt  s}  \rightarrow  {\tt  a}}  +  {\tt  Z}  \cdot  \theta_{{\tt  s}
\rightarrow  {\tt s  s}} $.
We can expand the solution {\tt Z} into an infinite series:
\begin{eqnarray*}
{\tt Z}
 & = &
       \frac{ 1 }{1 - \theta_{{\tt s}\rightarrow {\tt s s}} }
           {\tt W}\cdot  \theta_{{\tt s}\rightarrow {\tt a}} 
 \;\; = \;\; (1 + \theta_{{\tt s}\rightarrow {\tt s s}} + \theta_{{\tt s}\rightarrow {\tt s s}}^2
         + \cdots ){\tt W}\cdot  \theta_{{\tt s}\rightarrow {\tt a}}
\end{eqnarray*}

It is  easy to see that  this series represents the  probability of infinitely
many  leftmost derivations  of  prefix ``{\tt  a}''  from nonterminals  ``{\tt
  s\,s}'' by partitioning the derivations  based on the number of applications
of rule ${\tt  s}\rightarrow {\tt s\,s}$ to derive ``{\tt  a}'', i.e.\ $1$ for
no application (${\tt s\, s}\,  {\Rightarrow}_{ {\tt s}\rightarrow {\tt a} }\,
{\tt a\,  s}$), $\theta_{{\tt s}\rightarrow {\tt  s s}}$ for one  ( ${\tt s\,
  s}\, {\Rightarrow}_{  {\tt s}\rightarrow {\tt s  s} }\, {\tt  s\, s\, s\,}\,
{\Rightarrow}_{  {\tt  s}\rightarrow {\tt  a}  }\,  {\tt  a\,s\,s }$)  and  so
on\footnote{
Here we  use $\alpha {\Rightarrow} \beta$  (resp. $\alpha \derive  \beta$ ) to
indicate $\beta$ is derived from  $\alpha$ by one step derivation (resp.  zero
or  more steps  derivation) using  CFG  rules.  Also  recall here  that it  is
assumed that  PCFGs are consistent. So  the sum of  probabilities of sentences
derived from ``{\tt s}'' is 1.   Consequently for example we may safely ignore
{\tt s} in ``{\tt a s}''  when computing the probability of prefix ``{\tt a}''
derived from ``{\tt a s}''.
}.\\

\subsection{Properties of explanation graphs generated by a prefix parser}
\label{subsec:cyclicpcfg}
We here  examine properties of cyclic  explanation graphs.  Let {\bf  PG} be a
PCFG and  {\bf G'}  its underlying  CFG, i.e.\ the  CFG obtained  by removing
probabilities from {\bf PG}.  Throughout this subsection we use $\db_{\bf PG}$
for a  prefix parser  for {\bf  PG} obtained by  replacing the  {\tt values/2}
declaration in  $\db_0$ in Fig.~\ref{fig:prism:prog0} with  an appropriate set
of {\tt values/2}  declarations encoding {\bf PG}.  In  what follows, we first
prove  a  necessary and  sufficient  condition  under  which a  prefix  parser
$\db_{\bf  PG}$  generates cyclic  explanation  graphs.   We  then prove  that
$\db_{\bf  PG}$ always  generates  a  system of  linear  equations for  prefix
probabilities.  Finally we prove that  the linear system is solvable by matrix
operation under our assumptions on PCFGs.

\begin{theorem}
\label{th-1}
Let $G_{\ell}$ = {\tt pre\_pcfg($\ell$)} be  a goal for a prefix $\ell$ = {\tt
  [$w_1,\ldots, w_N$]} in {\bf G'} and $\expl(G_{\ell})$ an explanation graph
for  $G_{\ell}$ generated  by  $\db_{\bf  PG}$. Suppose  there  is no  useless
nonterminal  in  {\bf G'}.  Then  there  exists  a cyclic  explanation  graph
$\expl(G_{\ell})$  if-and-only-if the  left-corner  relation of  {\bf G'}  is
cyclic.
\end{theorem}
\noindent   
\begin{proof}
Suppose   $\expl(G_{\ell})$  is   cyclic.    Then  some   defined  goal   {\tt
  pre\_pcfg([$a$|$\beta$],$\ell_0$,$\ell_2$)} with  a nonterminal ``$a$'' must
have itself as  a descendant in $\expl(G_{\ell})$ where  $\ell_0$ and $\ell_2$
are sublists  of ${\ell}$.  So  an SLD derivation exists  from {\verb!:-!}{\tt
  pre\_pcfg([$a$|$\beta$],$\ell_0$,L2),K}  to  its descendant  {\verb!:-!}{\tt
  pre\_pcfg([$a$|$\beta$],$\ell_0$,L2'),K'}  in  which  the list  $\ell_0$  is
preserved.  Consequently there is a corresponding leftmost derivation ${\tt s}
\derive {a}\delta \derive {a}\delta'$ by  {\bf G'}, the underlying CFG of {\bf
  PG}.  So the left-corner relation is cyclic.

Conversely  suppose the  left-corner relation  of {\bf  G'} is  cyclic.  Then
there is a nonterminal ``$a$'' such that $a \rightarrow_{L}a$.  As there is no
useless nonterminal by our assumption, there is a leftmost derivation starting
from  ``{\tt   s}''  such  that  ${\tt  s}   \derive  {\gamma}a\delta  \derive
{\gamma}a\delta'  \derive w_1\ldots w_N$  for some  sentence $w_1,\ldots,w_N$.
In what  follows, for  simplicity we  assume that $\gamma$  is empty  (but the
generalization is  straightforward).  Let ${\ell_0}  = w_1,\ldots,w_j$ ($j\leq
N$) be a prefix derived from $a$  whose partial parse tree\footnote{
A partial  parse tree  is an  incomplete parse tree  whose leaves  may contain
nonterminals.
} has $a$ as  the root and no $a$ occurs below the root  $a$.  Then it is easy
to  see  that the  tabled  search for  all  proofs  of $G_{\ell_0}$  generates
$\expl(G_{\ell_0})$          containing          a          goal          {\tt
  pre\_pcfg([$a$|$\beta$],$\ell_0$,[])}  which is an  ancestor of  itself.  So
$\expl(G_{\ell_0})$ is cyclic.
\end{proof}

Let  $\expl(G)$ be  an  explanation  graph for  $G_{\ell}$.   We introduce  an
equivalence relation $A \equiv B$  over defined goals appearing in $\expl(G)$:
$A \equiv B$ if-and-only-if $A=B$ or $A$ is an ancestor of $B$ and vice versa.
We partition the set of defined goals into equivalence classes $[A]_{\equiv}$.
Each  $[A]_{\equiv}$  is  called  an  {\em  SCC\/}  ({\em  strongly  connected
  component}).   We  say  that  {\em  a defining  formula  $H  \Leftrightarrow
  \alpha_1 \vee\ldots\vee \alpha_M$  is linear\/} if there is  no $\alpha_i$ =
$C_1 \wedge\ldots\wedge C_m \wedge \mswz_1 \wedge\ldots\wedge \mswz_n$ ($1\leq
i \leq h$, $0 \leq m,n$) that  has two defined goals, $C_j$ and $C_k$ ($j \neq
k$), belonging to  the same SCC.  Also we say {\em  $\expl(G)$ is linear\/} if
every defining formula in $\expl(G)$ is linear.

\begin{lemma}
\label{lemma-3}
No two  defined goals in the  body of a defining  formula in $\expl(G_{\ell})$
belong to the same SCC.
\end{lemma}
\noindent   
\begin{proof}
Let $H \Leftrightarrow \alpha_1 \vee\ldots\vee \alpha_M$ be a defining formula
in  $\expl(G_{\ell})$.  Suppose  some  $\alpha_i$ contains  two defined  goals
belonging to the same  SCC.  Looking at $\db_0$ in Fig.~\ref{fig:prism:prog0},
we know  that the  only possibility is  such that $H  \Leftrightarrow \alpha_1
\vee\ldots\vee  \alpha_M$ is a  ground instantiation  of the  first (compound)
clause about {\tt pre\_pcfg/3}:

\begin{eqnarray}
\lefteqn{ \mbox{\tt pre\_pcfg([$a$|$\beta$],$\ell_0$,$\ell_2$)} \mbox{{\tt:-}}  } \nonumber \\
& &  \mbox{\tt  msw($a$,$\alpha$),pre\_pcfg($\alpha$,$\ell_0$,$\ell_1$),pre\_pcfg($\beta$,$\ell_1$,$\ell_2$)}
       \label{fig:groundcl}
\end{eqnarray}

\noindent
and the  two defined  goals, {\tt pre\_pcfg($\alpha$,$\ell_0$,$\ell_1$)  } and
{\tt  pre\_pcfg($\beta$,$\ell_1$,$\ell_2$)}, are  in the  same  SCC.  However,
since {\tt pre\_pcfg($\alpha$,$\ell_0$,$\ell_1$) }  is a proved goal, $\ell_1$
is    shorter   than    $\ell_0$.    On    the   other    hand    since   {\tt
  pre\_pcfg($\beta$,$\ell_1$,$\ell_2$)    }   is    an   ancestor    of   {\tt
  pre\_pcfg($\alpha$,$\ell_0$,$\ell_1$)  } in  $\expl(G_{\ell})$  because they
belong to the  same SCC by assumption,  $\ell_0$ is identical to or  a part of
$\ell_1$,  and  hence   $\ell_0$  is  equal  to  or   shorter  than  $\ell_1$.
Contradiction.    Therefore  there   is  no   such  defining   formula.  Hence
$\expl(G_{\ell})$ is linear.
\end{proof}

\begin{theorem}
\label{th-2}
Let $\expl(G_{\ell})$ be an explanation graph for a prefix $\ell$ parsed by
$\db_{\bf PG}$. $\expl(G_{\ell})$ is linear.
\end{theorem}

\noindent   
\begin{proof}
Immediate from Lemma~\ref{lemma-3}.
\end{proof}

We next  introduce a partial  ordering $[A]_{\equiv} \succ  [B]_{\equiv}$ over
SCCs by $[A]_{\equiv} \succ [B]_{\equiv}$ if-and-only-if $A$ is an ancestor of
$B$ but not vice versa in $\expl(G)$.  We then extend this partial ordering to
a  total ordering $[A]_{\equiv}  \succ [B]_{\equiv}$  over SCCs.   Likewise we
partition  P-variables  by  the  equivalence relation:  $P(A)  {\equiv}  P(B)$
if-and-only-if $[A]_{\equiv} =  [B]_{\equiv}$.  We denote by $[P(A)]_{\equiv}$
the  equivalence class  of  P-variables corresponding  to $[A]_{\equiv}$.   By
construction  $[P(A)]_{\equiv}$s are totally  ordered isomorphically  to SCCs;
$[P(A)]_{\equiv}  \succ  [P(B)]_{\equiv}$  if-and-only-if $[A]_{\equiv}  \succ
[B]_{\equiv}$.    In  the  following   we  treat   SCCs  and   P-variables  as
isomorphically    stratified    by    this    total    ordering.     We    use
$\eq([P(A)]_{\equiv})$ to stand for the union of sets of probability equations
for  defined goals  in  $[A]_{\equiv}$.  Notice  that  in the  case of  PCFGs,
$\eq([P(A)]_{\equiv})$ is  a system of linear  equations by Theorem~\ref{th-2}
if  we  consider   P-variables  in  the  lower  strata   as  constants.  Hence
${\eq(G_{\ell})}$ is solvable inductively from lower strata to upper strata.

Now we show that $\eq([P(A)]_{\equiv})$ is always solvable by matrix operation
under  our  assumptions  on  PCFGs.   Let  ``$a$'' be  a  nonterminal  in  the
underlying CFG {\bf  G'} and $A$ a defined  goal in $\expl(G_{\ell})$.  Write
$A$ = {\tt pre\_pcfg([$a$|$\beta$],$\ell_0$,$\ell_2$)}.  Since $A$ is a proved
goal,   $A$   successfully   calls   some   ground  goals   $B_{j}$   =   {\tt
  pre\_pcfg($\alpha_j$,$\ell_{0}$,$\ell_{1j}$)}    and    $C_{j}$    =    {\tt
  pre\_pcfg($\beta$,$\ell_{1j}$,$\ell_2$)}  in   the  clause  body   shown  in
(\ref{fig:groundcl})  where $a  \rightarrow \alpha_j$  is a  CFG rule  in {\bf
  G'}.  By  repeating a similar  proof for Lemma~\ref{lemma-3}, we  can prove
that  the  third goal  $C_{j}$  does not  belong  to  $[A]_{\equiv}$, the  SCC
containing      $A$.       Thus      $[A]_{\equiv}    \succ      [      \mbox{\tt
    pre\_pcfg($\beta$,$\ell_{1j}$,$\ell_2$)}   ]_{\equiv}$.    So  only   some
$B_{j}$s can possibly belong to $[A]_{\equiv}$.

Let   $P(A_1),\ldots,P(A_K)$    be   an   enumeration    of   P-variables   in
$[P(A)]_{\equiv}$.     Introduce    a    column    vector   ${\bf    X}_A    =
(P(A_1),\ldots,P(A_K))^T$.  It  follows from what we discussed  before that we
can write $\eq([P(A)]_{\equiv})$ as a  system of linear equations ${\bf X}_A =
M{\bf X}_A +  {\bf Y}_A$ where $M$  is a $K \times K$  non-negative matrix and
${\bf Y}_A$ is  a non-negative vector whose component is  a sum of P-variables
in the lower strata multiplied by constants.  $M$ is irreducible because every
goal   in  $[A]_{\equiv}$  directly   or  indirectly   calls  every   goal  in
$[A]_{\equiv}$  with positive  probability.  ${\bf  Y}_A$ is  non-zero because
some $A_i$  must have  a proof tree  that only  contains defined goals  in the
lower strata.\\

\begin{theorem}
\label{th-3}
Let {\bf PG} be a consistent PCFG such that there is no epsilon rule and every
production rule has a positive selection probability.  Also let $\db_{\bf PG}$
be a prefix parser for {\bf PG} and $\expl(G_{\ell})$ an explanation graph for
a  prefix  $\ell$.   Suppose  $\eq([P(A)]_{\equiv})$  is a  system  of  linear
equations for  a defined goal $A$ in  $\expl(G_{\ell})$.  Put $[P(A)]_{\equiv}
=\{ P(A_i) \mid 1 \leq i  \leq K \}$ and write $\eq([P(A)]_{\equiv})$ as ${\bf
  X}_A = M{\bf X}_A + {\bf Y}_A$ where ${\bf X}_A = (P(A_1),\ldots,P(A_K))^T$.
It has a unique solution ${\bf X}_A = (I-M)^{-1}{\bf Y}_A$.
\end{theorem}

\noindent
\begin{proof}
We prove  that $I-M$ has an inverse  matrix. To prove it,  we assume hereafter
that   P-variables  in   $[P(A)]_{\equiv}$  are   assigned  as   their  values
probabilities  from  ${\bf  X}^G_{\infty}$,  a  solution  for  $\eq(G)$  whose
existence  is guaranteed by   Theorem~\ref{th-0}  and  hence  all  equations  in
$\eq([P(A)]_{\equiv})$ are true.

By  applying ${\bf X}_A  = M{\bf  X}_A +  {\bf Y}_A$  $k$ times  repeatedly to
itself, we  have ${\bf X}_A =  M^{k} {\bf X}_A  + (M^{k-1}+\cdots+I){\bf Y}_A$
for  $k=1,2,\ldots$  Since   $M$,  ${\bf  X}_A$,  and  ${\bf   Y}_A$  are  all
non-negative, we  have ${\bf X}_A  \geq M^{k} {\bf  X}_A$ and ${\bf  X}_A \geq
(M^{k-1}+\cdots+I){\bf  Y}_A$ for every  $k$.  On  the other  hand since  $ \{
(M^{k-1}+\cdots+I){\bf Y}_A  \}_k$ is  a monotonically increasing  sequence of
non-negative  vectors bounded by  ${\bf X}_A$,  it converges  and so  does $\{
M^{k} {\bf X}_A \}_{k}$.

Let $\rho(M)$ be the spectral  radius of $M$\footnote{
$\rho(M) \ddefined \argmax{i} | \lambda_i |$
where the $\lambda_i$'s are the eigenvalues of  $M$.
}. Suppose $\rho(M)>1$.  In  general $\rho(M) \leq { {\parallel{M^k}\parallel}
}_{\infty}^{    \frac{1}{k}    }    $    holds    for    every    $k$    where
${{\parallel{\cdot}\parallel}}_{\infty}$ is  the matrix norm  induced from the
$\infty$     vector    norm.     It     follows    from     $\rho(M)^k    \leq
{{\parallel{M^k}\parallel}}_{\infty}$   that   $\lim_{k  \rightarrow   \infty}
{{\parallel{M^k}\parallel}}_{\infty}  =  +\infty$.   Consequently since  ${\bf
  X}_A >0$ holds because every proved goal has a positive probability from our
assumption,  some  element of  ${M^k}{\bf  X}_A$  goes  to $  +\infty$,  which
contradicts the convergence of $\{ M^{k} {\bf X}_A \}_k$. So $\rho(M) \leq 1$.

Suppose  now  $\rho(M)  = 1$.   Then  in  this  case,  we note  that  $\left\{
{\displaystyle  \frac{ M^{k-1}+\cdots+I  }{k} }  \right\}_{k}$ converges  to a
positive     matrix    (Example     8.3.2,    \cite{Meyer00})     and    hence
$(M^{k-1}+\cdots+I){\bf        Y}_A       =        \left(       {\displaystyle
  \frac{M^{k-1}+\cdots+I}{k}} \right)  \cdot k{\bf Y}_A$ diverges  as $k$ goes
to   infinity,    which   contradicts    again   the   convergence    of   $\{
(M^{k-1}+\cdots+I){\bf   Y}_A   \}_{k}$.   Therefore   $\rho(M)   <  1$.    So
$(I-M)^{-1}$ exists.
\end{proof}

Note that ${\bf  X}_A = (I-M)^{-1}{\bf Y}_A =  (I+M+M^2+\cdots){\bf Y}_A$.  By
further analyzing the matrix $M$,  we understand that multiplying $M$ by ${\bf
  Y}_A$ for  example corresponds  to growing partial  parse trees by  one step
application  of  production rules  (reduce  operation  in bottom-up  parsing).
Hence  $P(A_i)$,  a component  of  ${\bf X}_A$,  becomes  an  infinite sum  of
probabilities  and  so  is   the  probability  of  the  top-goal  $P(\mbox{\tt
  pre\_pcfg($\ell$)})$.\\

We  sum up our  discussion so  far and  state in  Fig.~\ref{fig:prefix_comp} a
general procedure to compute probability on cyclic explanation graphs.  In the
case of PCFGs,  $\db$ is the prefix parser  in Fig.~\ref{fig:prism:prog0} with
appropriate {\tt values/2}  declarations encoding a given PCFG  and $G$ = {\tt
  pre\_pcfg($\ell$)} is a goal for  a prefix $\ell$.  Under our assumptions on
PCFGs,  ${\eq(G)}$  in   {\bf  [Step  2]}  is  guaranteed   to  be  linear  by
Theorem~\ref{th-2}   and    {\bf   [Step    3]}   is   always    possible   by
Theorem~\ref{th-3}.

\begin{figure*}[ht]
\rule{\textwidth}{0.25mm}\\
\begin{description}
\item[{\bf [Step 1]:}] Given a program $\db$ and  a goal $G$, construct an explanation graph ${\expl(G)}$.
\item[{\bf [Step 2]:}] Convert ${\expl(G)}$ to a set of probability equations ${\eq(G)}$.
\item[{\bf [Step 3]:}] Solve  ${\eq(G)}$ inductively from lower strata
     by matrix operation and obtain $\pdb(G)$.
\end{description}
\rule{\textwidth}{0.25mm}\\
\caption{ Probability computation on cyclic explanation graphs }
\label{fig:prefix_comp}
\end{figure*}

We  emphasize  that the  procedure  is  general  and applicable  to  arbitrary
programs that generate linear explanation graphs\footnote{
Currently the  PRISM system  returns an error  message when ${\eq(G)}$  is not
linear.
}, not  restricted to those  generated by a  prefix PCFG parser.  Also  we add
that   even    if   ${\eq(G)}$   is   nonlinear,   it    is   still   solvable
(Theorem~\ref{th-0}).   This  fact is  applied  to  the  computation of  infix
probability for PCFGs~\cite{Nederhof11a} though it is beyond the scope of this
paper and we do not discuss it.

\subsection{Prefix probability computation for a real PCFG}
Here  we apply our  approach to  real data  to show  the effectiveness  of our
approach. We use the ATR corpus and its PCFG~\cite{Uratani94}\footnote{
All experiments in this  paper are done on a single machine  with Core i7 Quad
2.67GHz$\times$2  CPU and 72GB RAM running OpenSUSE 11.2.
%
%
}.   The corpus  contains labeled  parse trees  for 10,995  Japanese sentences
whose  average length  is about  10.   The associated  manually developed  CFG
comprises 861 CFG rules (168 non-terminals and 446 terminals\footnote{
In  this paper, we  use part-of-speech  (POS) tag  sequences derived  from the
sentences instead  of sentences themselves.   So terminals in the  grammar are
POS tags.
}) and yields 958 parses/sentence on average.  A PCFG is prepared by assigning
probabilities  ({\em parameter\/}s) to  CFG rules  and is  encoded as  a PRISM
program just like the  one in Fig.~\ref{fig:prism:prog0} with appropriate {\tt
  values/2} declarations.
Using this PCFG,  we computed the average probability of  sentence and that of
prefix in the ATR corpus for comparison.  We randomly sampled 100 sentences of
a given length from the ATR  corpus and computed their average probability. We
then  deleted their  last word  and  created 100  prefixes for  which we  also
computed the average probability.

Fig.~\ref{fig:prefix_atr} contains  results of plotting  the (minus) logarithm
of average prefix probability and  that of average sentence probability for a
length varying from  2 to 22.  We  used two parameter sets for  the PCFG.  For
the left  figure (a), parameters are  uniform, i.e.\ if a  nonterminal $X$ has
$n$ rules  $\{X \rightarrow \alpha_i \mid  1 \leq i  \leq n \}$, each  rule is
selected  with probability  $1/n$. For  the right  figure (b),  parameters are
learned from the entire ATR corpus by the built-in EM algorithm in PRISM.

\begin{figure}[hb]
\begin{tabular}{cc}
\includegraphics[scale=0.7]{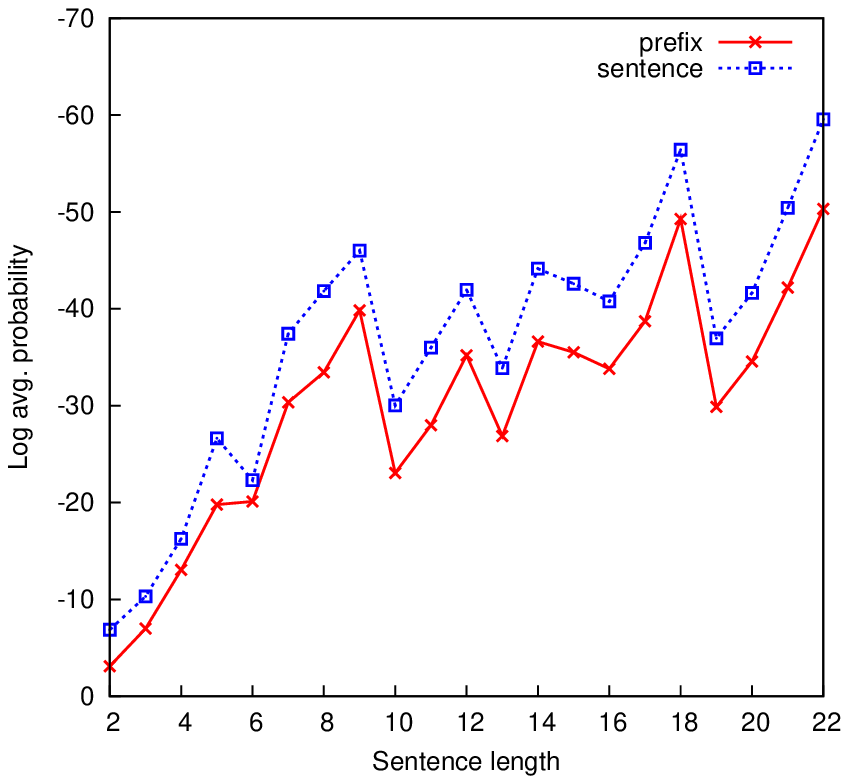}
&
\includegraphics[scale=0.7]{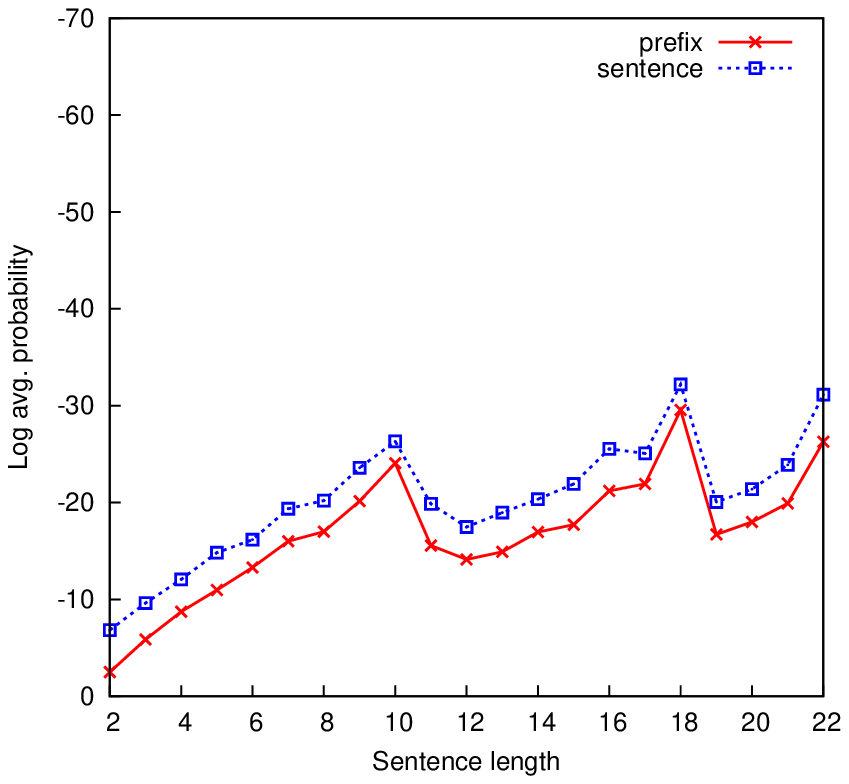} \\
\hspace{1em}(a) Uniform parameters & \hspace{1em}(b) Learned parameters
\end{tabular}
\caption{ Prefix and sentence probability for the ATR corpus }
\label{fig:prefix_atr}
\end{figure}

Seeing  these figures we  first note  that the  average prefix  probability is
always greater  than the average sentence  probability at each  length both in
(a) and in (b) as expected and second that the curves in (b) are much smoother
than the ones  in (a) and shifted downward considerably  (the y-axis is scaled
with  minus logarithm)  due to  the effect  of parameters  learned  by maximum
likelihood estimation.   It is also  observed that the difference  between the
two  curves in (b)  is smaller  than the  one in  (a), which  is statistically
confirmed by t-test at 0.05 significance level\footnote{
In (a), the  average difference between the two curves is  6.57 whereas in (b)
the average difference is 3.59.
}. 

One potential explanation for this phenomenon is as follows.  Let $uw$ be a
sentence in the ATR corpus where $w$ is the last word.  The difference between
the probability of prefix $u$ and  the probability of sentence $uw$ is the sum
of infinitely  many probabilities  of the sentences  $D$ extending  $u$ except
$uw$.  Since  most members  of $D$ do  not appear  in the corpus,  their total
probability computed  from the parameters  learned from the corpus  by maximum
likelihood  estimation considerably decreases  compared to  the case  of using
uniform parameters  where any one  of $D$ receives  non-negligible probability
mass.  Since this  happens to every prefix used in the  experiment, we see the
narrowing difference between the average sentence probability curve in (a) and
the average prefix probability curve in (b).

One of  the usage  of prefix  probability computation is  to predict  the most
likely next word of a prefix  $u$.  Let $P_{\rm cfg}(\cdot)$ be a distribution
over sentences  by a  PCFG.  Then the  {\em conditional  prefix probability\/}
$P_{\rm prefix}(w \mid u)$ of a word $w$ given $u$ is computed as
$P_{\rm  prefix}(w \mid u)
    \ddefined {\displaystyle \frac{P_{\rm prefix}(uw)}{P_{\rm prefix}(u)} }$
where ${\displaystyle P_{\rm prefix}(u) = \sum_{uv : {\rm sentence}} P_{\rm cfg}(uv) }$.

Since  we   found  the  prefix  probability   computation  is  computationally
burdensome for  long prefixes, we tested  short prefixes.  For  example, for a
prefix  $u  =  {\tt  [t\_interj\_hesit,t\_interj\_pre,t\_daimeisi\_domo]}$  of
length three  and a  word $w =  {\tt t\_myoji\_first}$, we  calculated
$P_{\rm prefix}(w  \mid  u)$,  assuming  equiprobable  rule  selection,  as
$P_{\rm  prefix}(w \mid  u) = \mbox{0.00103}$.   Thus by computing
$P_{\rm prefix}(w\mid u)$ for  all possible $w$s, we  can predict the most likely  next word of
the given prefix as $\argmax{w} P_{\rm prefix}(w \mid u)$.

\section{Prefix probability computation for PLCGs}
\label{sec:preplc}
In this section,  we deal with probabilistic  left-corner grammars (PLCGs) and
their prefix probability computation to test the generality of our approach.

PLCGs are a probabilistic version of left-corner grammars (LCGs) which in turn
are       a      generative       version       of      left-corner       (LC)
parsing~\cite{Manning97,Roark99,Uytsel01}  that   performs  bottom-up  parsing
using  three parsing  operations, i.e.\  shift, attach  and  project. Although
PLCGs and PCFGs  may share a common CFG,  they assign probability differently.
PCFGs  assign probability to  the expansion  of nonterminals  by CFG  rules in
top-down parsing whereas  PLCGs assign probability to the  three operations in
bottom-up parsing.  As a result they define different classes of distribution.\\

\begin{figure*}[t]
\rule{\textwidth}{0.25mm}\\ [-1em]
\begin{verbatim}
values(lc(s,s),[rule(s,[s,s])]).  values(lc(s,a),[rule(s,[a])]).
values(lc(s,b),[rule(s,[b])]).    values(first(s),[a,b]).
values(att(s),[att,pro]).

pre_plcg(L):- g_call([s],L,[]).     % L is a prefix
g_call([],L,L).
g_call([G|R],[Wd|L],L2):-
   ( G = Wd -> L1 = L               % shift operation
   ; msw(first(G),Wd),lc_call(G,Wd,L,L1) ),
   ( L1 == [] -> L2 = []            % (pseudo) success
   ; g_call(R,L1,L2) ). 
lc_call(G,B,L,L2):-                 % B-tree is completed   
   msw(lc(G,B),rule(A,[B|RHS2])),
   ( G == A -> true ; values(lc(G,A),_) ),
   ( L == [] -> L1 = []             % (pseudo) success
   ; g_call(RHS2,L,L1) ),
   ( G == A -> att_or_pro(A,Op),    % attach or project
     ( Op == att -> L2 = L1 ; lc_call(G,A,L1,L2) )
   ; lc_call(G,A,L1,L2) ).
att_or_pro(A,Op):- ( values(lc(A,A),_) -> msw(att(A),Op) ; Op=att ).
\end{verbatim}
\rule{\textwidth}{0.25mm}\\ 
\caption{ Prefix PLCG  parser $\db_1$ }
\label{fig:prism:prog1}
\end{figure*}

Since prefix probability  computation for PLCGs does not  seem to be attempted
before,    we    detail   how    a    prefix    PLCG    parser   $\db_1$    in
Fig.~\ref{fig:prism:prog1} works. It is a  serial parser and specialized for a
PLCG  whose underlying  CFG  is ${\bf  G}_0$  = \{  {\tt s}$\rightarrow$  {\tt
  s}\,{\tt s},  {\tt s}$\rightarrow$ {\tt a}, {\tt  s}$\rightarrow${\tt b} \},
the same as the one for $\db_0$ in Subsection~\ref{subsec:exampleprepcfg}.
In  the program {\tt  values(lc($g$,$b$),$r$)} introduces  {\tt msw}  atoms to
choose a CFG rule $g \rightarrow b\beta$ from $r$ where $g$ and $b$ are in the
left-corner relation of ${\bf G}_0$.  So {\tt values(lc(s,s),[rule(s,[s,s])])}
introduces        just       one        {\tt       msw}        atom       {\tt
  msw(lc(s,s),[rule(s,[s,s])])}\footnote{
Consequently  executing   {\tt  msw(lc(s,s),X)}   returns  ${\tt  X}   =  {\tt
  rule(s,[s,s])}$ with probability 1.
}. On the  other hand {\tt values(first(s),[a,b])} that  encodes the first set
of ``{\tt s}'' in ${\bf G}_0$\footnote{
The first set of a nonterminal $A$  is the set of terminals in the left-corner
relation with $A$.
} introduces  $\{ \mbox{\tt msw(first(s),a)},  \mbox{\tt msw(first(s),b)} \}$.
Similarly    {\tt   values(att(s),[att,pro])}    introduces    $\{   \mbox{\tt
  msw(att(s),att)},  \mbox{\tt msw(att(s),pro)}  \}$ to  make  a probabilistic
choice between attach and  project.  All probabilistic choices are equiprobable
by default.

Suppose  {\tt pre\_plcg($\ell$)}  is given  as a  top-goal where  $\ell$  is a
prefix.   To  parse $\ell$,  the  parser  repeatedly  performs shift  by  {\tt
  g\_call/3} and attach and project by {\tt lc\_call/4} just as in LC parsing.
The role of {\tt g\_call($\alpha$,$\ell$,L2)}  is to construct a partial parse
tree  whose leaves  are a  substring {\tt  $\ell$-L2} (as  d-list)  spanned by
$\alpha$ while instantiating {\tt L2} to  a sublist of $\ell$.  Let {\tt G} be
the left-most  symbol of $\alpha$ and  {\tt Wd} the left-most  word of $\ell$.
When {\tt G} is a terminal and coincides with {\tt Wd}, shift is performed and
{\tt Wd}  is read from $\ell$ as  an initial partial parse  tree consisting of
{\tt Wd}.  Otherwise  {\tt Wd} is considered as a  word randomly selected from
the first  set of {\tt G}  using {\tt msw(first(G),Wd)} as  an initial partial
parse tree.

\begin{figure*}[t]
\rule{\textwidth}{0.25mm}\\ [-1em]
\begin{verbatim}
pre_plcg([a,b]) <=> g_call([s],[a,b],[])
g_call([s],[a,b],[]) <=> lc_call(s,a,[b],[]) & msw(first(s),a)
lc_call(s,a,[b],[])
   <=> g_call([],[b],[b]) & att_or_pro(s,pro)
          & lc_call(s,s,[b],[]) & msw(lc(s,a),rule(s,[a]))
g_call([],[b],[b])
lc_call(s,s,[b],[])
   <=> g_call([s],[b],[]) & att_or_pro(s,att)                  --(1)
          & msw(lc(s,s),rule(s,[s,s]))
     v g_call([s],[b],[]) & att_or_pro(s,pro)                  --(2)
          & lc_call(s,s,[],[]) & msw(lc(s,s),rule(s,[s,s]))
g_call([s],[b],[]) <=> lc_call(s,b,[],[]) & msw(first(s),b)    --(3)
lc_call(s,b,[],[])
   <=> att_or_pro(s,att) & msw(lc(s,b),rule(s,[b]))
     v att_or_pro(s,pro) & lc_call(s,s,[],[])
          & msw(lc(s,b),rule(s,[b]))
lc_call(s,s,[],[])
   <=> att_or_pro(s,att) & msw(lc(s,s),rule(s,[s,s]))
     v att_or_pro(s,pro) & lc_call(s,s,[],[])
          & msw(lc(s,s),rule(s,[s,s]))
att_or_pro(s,att) <=> msw(att(s),att)
att_or_pro(s,pro) <=> msw(att(s),pro)
\end{verbatim}
\rule{\textwidth}{0.25mm}\\ 
\caption{ Explanation graph for {\tt pre\_plcg([a,b])}  }
\label{fig:prism:probf1}
\end{figure*}

A call to  {\tt lc\_call(G,B,L,L2)} occurs when a  {\tt B}-tree (partial parse
tree  whose root  node  is {\tt  B})  is constructed  and {\tt  G}  is in  the
left-corner   relation   with   {\tt   B}.    It  grows   the   {\tt   B}-tree
probabilistically either by attach or by  project using a CFG rule of the form
{\tt A}$\rightarrow${\tt  B}$\beta$ until a {\tt G}-tree  is constructed while
consuming words in {\tt L}, leaving {\tt L2}.  When the input is a prefix, the
parser  returns with  pseudo success  as  soon as  the prefix  is consumed  as
indicated by the comment ``{\verb!(pseudo) success!}''.\\

When  {\tt pre\_plcg([a,b])}  is given  as a  top-goal for  example,  a linear
explanation graph shown in Fig.~\ref{fig:prism:probf1} is constructed in which
{\tt lc\_call(s,s,[],[])} calls itself as a top-most looping subgoal.
Now we  analyze Fig.~\ref{fig:prism:probf1} to  confirm that our  PLCG program
correctly  recognizes  all  partial  parse  trees  for  prefix  ``{\tt  ab}''.
Fig.~\ref{fig:prism:probf1}  compactly represents as  a form  of propositional
PRISM program  all computation paths (sequences of  probabilistic choices made
by {\tt msw}  atoms) that generate prefix ``{\tt  ab}''. Each path corresponds
to a partial  parse tree for ``{\tt ab}''.  We write  partial parse trees like
$\mbox{\tt s(s(s(a),s(b)),s)}$.  We denote by ${\bf T}_i$ ($i=1,2,3$) a set of
partial parse trees generated by  computation paths corresponding to line {\tt
  ($i$)} in Fig.~\ref{fig:prism:probf1}.

Then observe for example that computation paths going through {\tt (1)} yield
partial parse  trees by combining ones generated  by {\tt g\_call([s],[b],[])}
and  ones obtained  by  {\tt s}-trees  grown  by attach  operation using  rule
\mbox{{\tt s} $\rightarrow$  {\tt s}\,{\tt s}}.  This observation  leads to an
equation  \mbox{${\bf T}_1$}  =  {\tt s(s(a),\mbox{${\bf  T}_3$})} where  {\tt
  s(s(a),${\bf T}_3$)}  stands for the set $\{  \mbox{\tt s(s(a),$\tau$)} \mid
\tau \in {\bf T}_3 \}$.  In this way we obtain three equations below.
\begin{eqnarray*}
\mbox{\tt Eq 1:}\hspace{1em}\mbox{${\bf T}_1$} 
   & = & {\tt s(s(a),\mbox{${\bf T}_3$} )}  \\
\mbox{\tt Eq 2:}\hspace{1em}\mbox{${\bf T}_2$}
   & = & {\tt s( \mbox{${\bf T}_1$},s)}
        \cup {\tt s(\mbox{${\bf T}_2$},s)}\\
\mbox{\tt Eq 3:}\hspace{1em}\mbox{${\bf T}_3$}
   & = & \{ {\tt s(b)} \} 
        \cup {\tt s(\mbox{${\bf T}_3$},s)}
\end{eqnarray*}

By solving  them   we  know   that
${\bf  T}_3 =  
\{\hspace{0.3em}
   \overbrace{{\tt s(}\cdots{\tt s(}}^m{\tt s(b),}
   \overbrace{{\tt s)}\cdots{\tt s)}}^m
                   \mid m\geq 0 
\hspace{0.3em}\}$
and so on.   Also recall that all computation paths for  ``{\tt ab}'' have to
prove {\tt  lc\_call(s,s,[b],[])} and  hence have to  go through {\tt  (1)} or
{\tt  (2)} in  Fig.~\ref{fig:prism:probf1}.  Consequently  the set  of partial
parse trees  for prefix  ``{\tt ab}'' generated  by $\db_1$ is  represented as
${\bf T}_1 \cup {\bf T}_2$ where
${\bf  T}_1  \cup {\bf T}_2 =
\{\hspace{0.3em}
    \overbrace{{\tt s(}\cdots{\tt s(}}^n
       {\tt s(s(a),}
              \overbrace{{\tt s(}\cdots{\tt s(}}^m{\tt s(b),}
              \overbrace{{\tt s)}\cdots{\tt s)}}^m
            {\tt )}
     \overbrace{{\tt s)}\cdots{\tt s)}}^n
                  \mid m\geq 0, n\geq 0
\hspace{0.3em}\}$
which  certainly represents all partial  parse trees  for prefix ``{\tt  ab}''.\\

The probability  equations derived from  Fig.~\ref{fig:prism:probf1} are shown
in       Fig.~\ref{fig:prism:probeq1}.        We       have       $P(\mbox{\tt
  g\_call([],[b],[b])})=1$ as  {\tt g\_call([],[b],[b])} is  logically proved.
Suppose   the   probabilities   of   \mbox{\tt   msw(att(s),att)},   \mbox{\tt
  msw(att(s),pro)}, \mbox{\tt  msw(first(s),a)} and \mbox{\tt msw(first(s),b)}
are all  set to 0.5. Then  the solution becomes {\tt  X1} = {\tt  X2} = 0.125,
{\tt X3} = 0.25, {\tt X4} = 1, {\tt  X5} = {\tt X6} = 0.5, {\tt X7} = {\tt X8}
=  1  and  {\tt  X9}  =  {\tt   X10}  =  0.5.   So  the  probability  of  {\tt
  pre\_plcg([a,b])} is computed as {\tt X1} = 0.125.\\

\begin{figure*}[t]
\rule{\textwidth}{0.25mm}\\[0.6em]
$
\begin{array}{rlrll}
P( {\tt pre\_plcg([a,b])} )      & = & {\tt X1} & = & {\tt X2}  \\
P( {\tt g\_call([s],[a,b],[])} ) & = & {\tt X2} & = & {\tt X3}\cdot\mbox{0.5} \\
P( {\tt lc\_call(s,a,[b],[])} )  & = & {\tt X3} & = &
      {\tt X4}\cdot{\tt X10}\cdot{\tt X5}\cdot{1}                \\
P( {\tt g\_call([],[b],[b])} )   & = & {\tt X4} & = & 1          \\
P( {\tt lc\_call(s,s,[b],[])} )  & = & {\tt X5} & = &
     {\tt X6}\cdot{\tt X9}\cdot{1} + {\tt X6}\cdot{\tt X10}\cdot{\tt X8}\cdot{1} \\
P( {\tt g\_call([s],[b],[])} )   & = & {\tt X6} & = & {\tt X7}\cdot\mbox{0.5}    \\
P( {\tt lc\_call(s,b,[],[])} )   & = & {\tt X7} & = &
     {\tt X9}\cdot{1} + {\tt X10}\cdot{\tt X8}\cdot{1}        \\
P( {\tt lc\_call(s,s,[],[])} )   & = & {\tt X8} & = &
     {\tt X9}\cdot{1} + {\tt X10}\cdot{\tt X8}\cdot{1}        \\
P( {\tt att\_or\_pro(s,att)} )  & = &  {\tt X9} & = & \mbox{0.5}  \\
P( {\tt att\_or\_pro(s,pro)} ) & = & {\tt X10} & = & \mbox{0.5}  \\[0.6em]
\end{array}
$
\rule{\textwidth}{0.25mm}\\ 
\caption{ Probability equations for prefix ``{\tt ab}'' }
\label{fig:prism:probeq1}
\end{figure*}

Finally we test  prefix probability computation for PLCGs  with real data.  We
prepared  a prefix  PLCG  parser like  the  one in  Fig.~\ref{fig:prism:prog1}
adapted  for the  ATR  corpus and  conducted  prefix probability  computation.
Since the prefix PLCG parser is much larger than the corresponding prefix PCFG
parser, containing over 20,000  {\tt values/2} declarations, learning time and
computation time are expected to be much longer than the PCFG case.
Indeed, we  measured CPU time for the  PCFG and the PLCG  respectively used to
compute the  probabilities of 100 prefixes  created from 100  sentences in the
ATR corpus  by deleting their  last word.  The  PCFG case took  12.3 ms/prefix
whereas the PLCG case took 5.9 sec/prefix, 48 times slower than the PCFG case.
We also  computed conditional probability  $P_{\rm prefix}(w \mid u)$  for PLCG
prefixes.      For    a     pair     of    the     prefix     $u    =     {\tt
  [t\_interj\_hesit,t\_interj\_pre,t\_daimeisi\_domo]}$ and the word $w = {\tt
  t\_myoji\_first}$  used in  Section~\ref{sec:prepcfg}  for example,  $P_{\rm
  prefix}(w  \mid u)$  is computed  as  0.00032 which  is considerably  smaller
compared to 0.00103 computed for the PCFG case.

\section{Plan recognition}
\label{sec:plan_recog}
Prefix probability  computation has practical applications.   In this section,
we apply it to plan recognition  using artificial data.  Plan recognition is a
task of inferring  a plan (intension) from a sequence  of observed actions and
has been  pursued for example  in robotics to  interpret video scene  data and
sensor data.  One  way to perform plan recognition is to  use a formal grammar
to  describe the  relation  between  plans and  action  sequences by  equating
sentences with action  sequences and nonterminals with plans.  However to cope
with noisy observations,  it is natural to use  probabilistic grammars such as
PCFGs~\cite{Bobick98,Lymberopoulos07,Amft07,Geib11,Pomponio11}.\\

Consider  a simple  PCFG in  Fig.~\ref{tbl:plan_recog_1}  where S  is a  start
symbol.   It describes  how  four plans,  i.e.\  \{ Pl(playing),  St(studying),
Cl(cleaning), Mo(mowing) \} generate sequences of observable actions, i.e.\ \{
play, study, clean, mow \}.

\begin{figure}[h]
\rule{\textwidth}{0.25mm}\\ [0.2em]
\begin{tabular}{l@{$\ \rightarrow\ $}lrlrlrlrl}
   S & Pl    &: 0.1 $|$& St       &: 0.4 $|$& Cl    &: 0.3 $|$& Mo    &: 0.2 \\
   Pl& play  &: 0.5 $|$& play Pl  &: 0.3 $|$& Cl    &: 0.1 $|$& Mo    &: 0.1 \\
   St& study &: 0.1 $|$& study St &: 0.3 $|$& Pl St &: 0.2 $|$& Cl St &: 0.4 \\
   Cl& clean &: 0.4 $|$& clean Cl &: 0.5 $|$& Pl Cl &: 0.1 $ $&              \\
   Mo& mow   &: 0.3 $|$& mow Mo   &: 0.1 $|$& Pl Mo &: 0.4 $|$& Cl Mo &: 0.2   
\end{tabular}\\ 
\rule{\textwidth}{0.25mm}
\caption{ PCFG for plan recognition }
\label{tbl:plan_recog_1}
\end{figure}

This  PCFG generates  action sequences  such as  ``play clean'',  ``play study
study'' and so on.  Note that  although ``play clean'' is a sentence derivable
from Pl and Cl, it is also derivable  from Cl and Mo as a prefix.  Our task is
to predict, given  such a sequence $x$  of actions which may be  a prefix, the
most likely plan $y^* = \argmax{y} P_{\rm prefix}({\rm S} \rightarrow y, y
\derive x)$ where  $y$ ranges over \{ Pl,  St, Cl, Mo \} as  a recognized plan
for $x$.  For example, for $x$ =  ``play clean'', $y^*$ = St is the recognized
plan  giving  the  highest  probability  0.0272  for  $P_{\rm  prefix}({\rm  S}
\rightarrow y, y \derive x)$.

To evaluate the accuracy of our  prediction method, we take a random sample of
100 prefixes  (action sequences)  together with their  plans and  evaluate the
accuracy  of  prediction  by  predicting  the plan  for  each  sampled  action
sequence.  Prefixes are sampled so that they are not shorter than a threshold,
referred to hereafter as ``minlen''  (minimum length).  Finally we compute the
average accuracy  of prediction over  10 runs.  We  tested two cases.   One is
full observation, i.e.\ prefixes are  restricted to sentences.  The other case
is  no restriction.   Fig.~\ref{fig:plan_recog_1} shows  the  average accuracy
w.r.t.\  minlen varying  from 1  to 15.   The blue  curve corresponds  to full
observation   (sentence)   whereas  the   red   one   corresponds  to   prefix
observation\footnote{
We  measured the CPU  time for  plan recognition  with randomly  generated 100
action sequences  whose average length is  4.18 (with std  3.02).  We obtained
1.79ms/action sequence as the average time for plan recognition.
}.

\begin{figure}[h]
\begin{minipage}{0.3\hsize}
\centerline{
\includegraphics[scale=0.85]{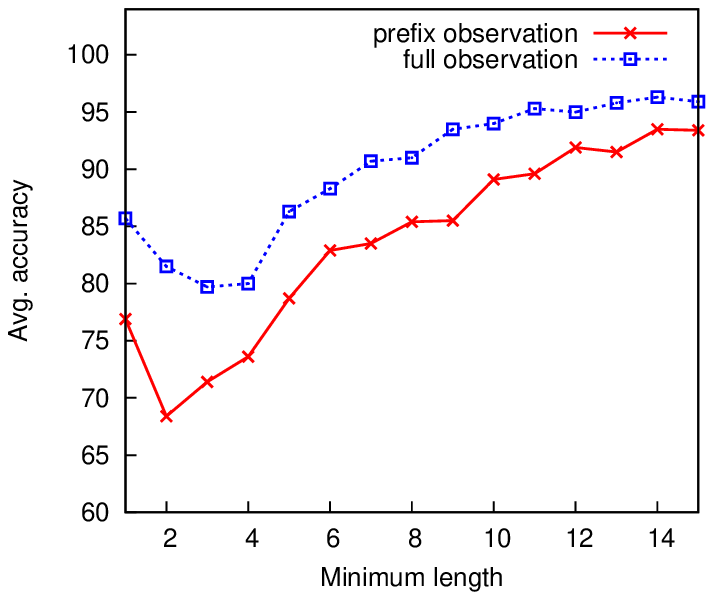}
}
\end{minipage}
\caption{ Plan recognition }
\label{fig:plan_recog_1}
\end{figure}

We  first notice that  full observation  always gives  a better  accuracy than
prefix  observation.   This may  be  attributed  to  the fact  that  ambiguity
measured  by the  average number  of possible  plans for  an  action sequence,
termed ``amb'' here, in  the case of full observation is less  than the amb in
the case of prefix observation at all  minlen values (1 to 15). We also observe
that the average accuracy (almost) monotonically increases as minlen increases
in both  cases.  This is  intuitively obvious because longer  action sequences
should give  more clue to prediction  and reduce the  ambiguity about possible
plans.   Actually amb monotonically  decreases w.r.t.\  minlen.  On  the other
hand, however, this explanation conflicts with the initial drop in both curves
w.r.t\ minlen, so we still need a coherent explanation.

\section{Reachability probability}
\label{sec:reachability}
Computing  probability  through  cyclic  explanation graphs  has  applications
beyond  prefix   probability  computation.   In  this   section,  inspired  by
\cite{Gorlin12}, we take up  the problem of computing reachability probability
in discrete Markov  chains.  Fig.~\ref{fig:reachable_1} illustrates an example
of  Markov  chain  (the  left-hand  side,  (a)) and  its  PRISM  program  (the
right-hand  side,  (b)), both  borrowed  from  \cite{Gorlin12}  with a  slight
modification of the program.

\begin{figure*}[h]
\begin{tabular}{cc}
\includegraphics[scale=0.75]{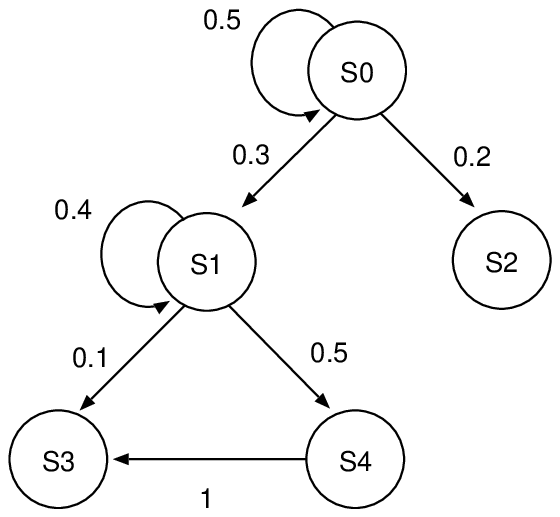}
&
\hspace{2em}
\includegraphics[scale=0.58]{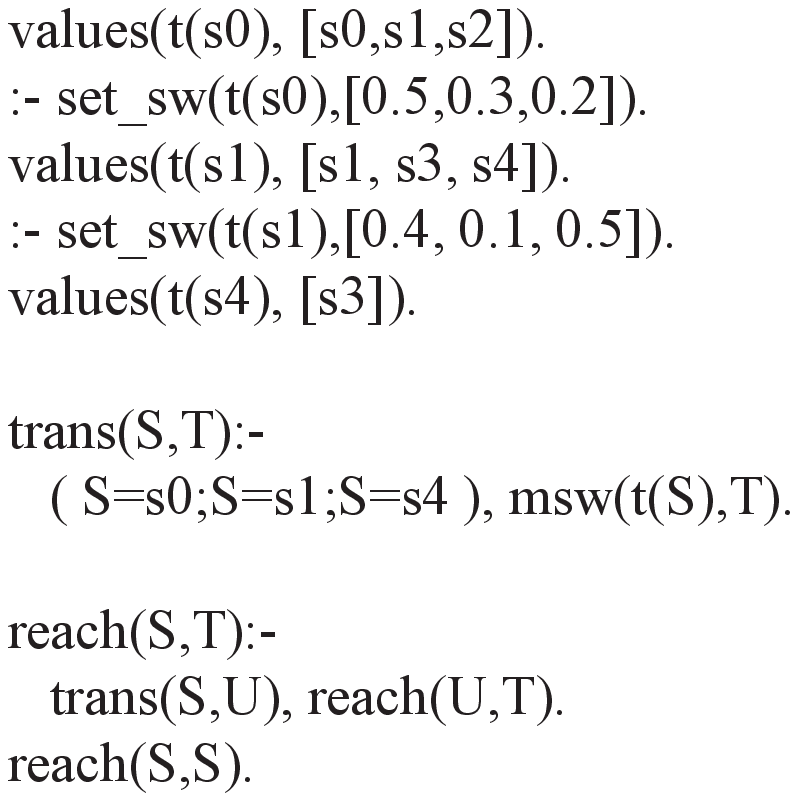} \\
(a)  &  (b)
\end{tabular}
\caption{ Markov chain (a) and a program (b) }
\label{fig:reachable_1}
\end{figure*}

A state transition in a Markov chain is made by a probabilistic choice of next
state. Since the choice is exclusive  and independent at each state, PRISM can
simulate  Markov  chains, except when there  is  a  self-loop, or  more
generally there  is a set of state  transitions forming a loop.   In this case
probability computation requires an  infinite sum of probabilities which PRISM
has  been unable to  deal with.   However, by  applying the  general procedure
described  in  Fig.~\ref{fig:prefix_comp},  we  are  now able  to  compute  an
infinite sum of probabilities,  in particular for the reachability probability
problem.  For example, the reachable probability  from {\tt s0} to {\tt s3} is
represented as P({\tt reach(s0,s3)}) and is computed by the program as 0.6.

In  the  following  we  tackle  a  more complicated  problem  and  verify  the
Synchronous Leader Election Protocol as described in a web page\footnote{
{\tt http://www.prismmodelchecker.org/casestudies/synchronous\_leader.php}
} for the PRISM model checker~\cite{Kwiatkowska11} as one of the case studies.
The protocol  probabilistically elects  a leader among  processors distributed
over a ring network communicating  by synchronous message passing.  It has two
parameters, N,  the number of  processors and K,  the number of  candidate ids
used for  election.  Our task is  to show that  a leader will be  elected with
probability one.   We use a PRISM  program faithfully translated  from the one
shown in the  web page with one exception. That  is, we separate probabilistic
transition from  deterministic transition and only  the predicate representing
the former is tabled using PRISM's {\tt p\_table} declaration\footnote{
{\tt :-  p\_table q/n} implies that  the probabilistic predicate  {\tt q/n} is
tabled {\em  and\/} all  other probabilistic predicates  not declared  by {\tt
  p\_table} will not  be tabled.  By default all  probabilistic predicates are
tabled in PRISM, which  sometimes makes explanation graphs unnecessarily large
in view of probability computation due to the introduction of defining clauses
without  {\tt  msw}s  in  the  body.   Selective  tabling  by  {\tt  p\_table}
declarations prevents this.
}.

\noindent
\begin{figure}[h]
\centerline{\includegraphics[scale=0.7]{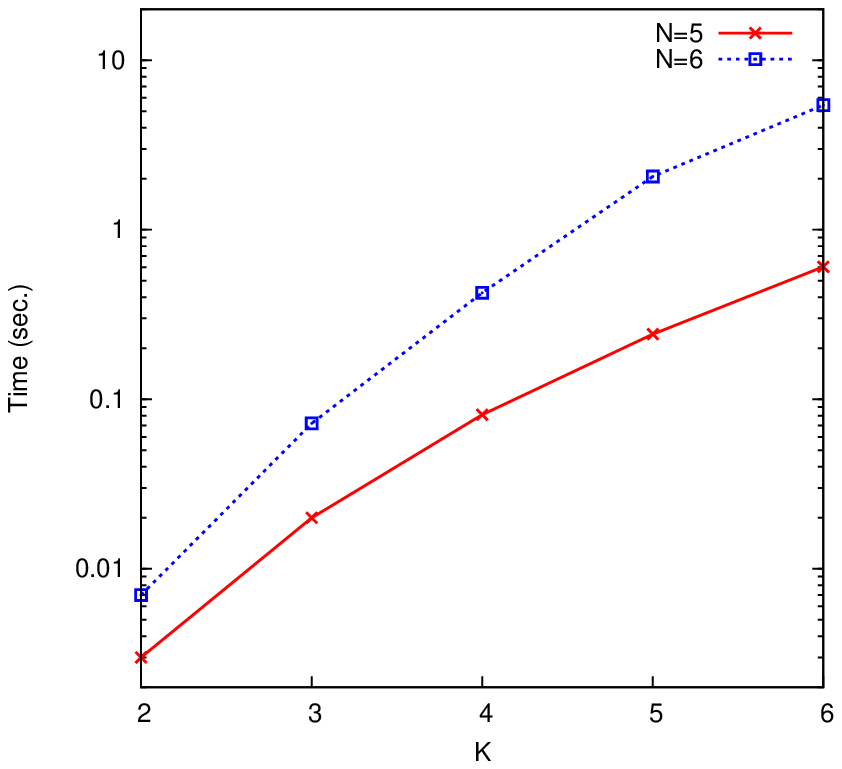}}
\caption{ CPU time for checking the Synchronous Leader Election Protocol }
\label{fig:mck}
\end{figure}

Fig.~\ref{fig:mck} shows CPU time taken for verification, varying N and K.  As
we see,  the plotted curves  for N=5 and  N=6 look alike  and the CPU  time is
almost exponential in K.   We note that they are similar in  shape to the ones
(PIP-full) obtained  by \citeN{Gorlin12} who conducted the  same experiment to
compare  their  approach  with  the  PRISM model  checker.  However  an  exact
comparison with our  approach would be difficult because  of the difference in
CPU   processors   and  more   seriously   because   of   the  difference   in
coding\footnote{
In (Gorlin  et al. 2012) the authors  used a 2.5GHz processor  and encoded the
Synchronous Leader Election Protocol problem  via a PCTL model checker whereas
we  used a  2.67GHz processor  and  directly encoded  the problem  as a  PRISM
program.
}.

\section{Related work and future work}
\label{sec:relatedwork}
Tabling  in  logic  programming  has  traditionally  been  used  to  eliminate
redundant computation and to avoid infinite loop, but the use of loop detected
by tabling for infinite probability  computation seems new, though tabling for
(finite)  probability  computation  is  well-known  and  implemented  in  some
probabilistic  logic  programming   languages  such  as  PRISM~\cite{Sato01g},
ProbLog~\cite{Mantadelis10}  and   PITA~\cite{Riguzzi11}.   This  is  probably
because looping goals have long  been considered useless despite the fact that
they make sense if probabilities are involved and the loop computes converging
probabilities like prefix probability computation.\\

Technically our approach is closely related to \cite{Gorlin12} in which Gorlin
et  al.\  proposed  PIP  (probabilistic  inference  plus)  that  computes  the
probability of infinitely many explanations and applied PIP to model checking.
In PIP, to compute the probability  of a query Q w.r.t.\ a probabilistic logic
program   P,   a   residual   program   is   first   constructed   using   XSB
Prolog~\cite{Swift12} from  P and  Q.  Then  it is converted  to a  DCG called
equation generator that generates possible explanations for Q as strings, from
which a  factored explanation  diagram (FED) is  derived.  It is  a compressed
representation of  the set  of (possibly infinitely  many) explanations  for Q
w.r.t.\  P and further  converted to  a system  of polynomial  equations.  The
probability of Q is obtained by solving the equations.

The basic idea  of PIP is similar to our  approach: probability computation by
solving a set of equations derived  from a symbolic diagram constructed from a
program and a query. Nonetheless there are substantial differences between PIP
and our  approach.  First  PIP uses  {\tt msw/3} that  has three  arguments in
which  the  second  argument  (trial-id  \cite{Sato01g})  is  a  term  (clock)
indicating when the {\tt msw} is executed in the computing process.  To ensure
statistically  correct  treatment  of  the  second  argument  for  probability
computation, PIP requires programs to be ``temporally well-formed'' and places
three  syntactic  conditions on  the  occurrences  of ``instance  arguments'',
i.e.\ arguments that  work as a clock.  These  conditions look restrictive but
how they  affect the class of  definable probabilistic models or  how they are
related to PRISM programs is unclear and not discussed in \cite{Gorlin12}.

PRISM on  the other hand uses {\tt  msw/2} that omits the  second argment from
{\tt  msw/3} for computational  efficiency and  allows arbitrary  programs but
instead assumes  every occurrence of {\tt msw/2}  in a proof for  the query is
independent ({\em independence condition}) which guarantees the correctness of
probability computation in PRISM.

Also PIP constructs an FED, BDD-like graphical structure representing a set of
explanations  via  a DCG  (equation  generator)  whereas  PRISM constructs  an
explantion graph without  using a DCG.  FEDs are powerful;  they enable PIP to
deal with programs that violate  the exclusiveness condition required by PRISM
while  capturing common  patterns in  the  set of  explanations. However  when
programs satisfy  the exclusiveness condition (and  the independence condition
as well) as  is often the case in probabilistic  modeling by generative models
such as BNs,  HMMs, PCFGs and PLCGs, the construction  of FEDs is unnecessary.
A simpler structure, explanation graphs,  is enough.  As we have demonstrated,
the sum  of probablities  of infinitely many  explanations can  be efficiently
computed by cyclic explanation graphs in such cases.

In addition,  though it  is not clearly  stated in \cite{Gorlin12},  Gorlin et
al.\ seem  to solve the set of  equations by an iterative  method described in
\cite{Etessami09} that is applicable  to nonlinear cases.  PRISM contrastingly
assumes the linearity of equations and efficiently solves hierarchally ordered
sets of system of linear equations, corresponding to SCCs, by matrix operation
in  cubic  time  in  the  number  of variables.   Considering  the  fact  that
nonlinearity  occurs  even  in  the  case  of  PCFGs  when  we  compute  infix
probability~\cite{Nederhof11a} however, it is important future work to enhance
PRISM's equation solving ability for nonlinear cases.

Current  tabling  in  PRISM employs  linear  tabling  in  B-Prolog and  it  is
straightforward to  construct cyclic explanation graphs  from defining clauses
for  tabled answers  stored in  the memory.   Constructing  cyclic explanation
graphs  in other  Prolog  systems  such as  XSB~\cite{Swift12}  that employ  a
suspend-resume mechanism for tabling also seems possible.

Approximate computation  of prefix probability  seems possible for  example by
the iterative deepening algorithm used in ProbLog \cite{DeRaedt07}. To develop
such an approximation algorithm remains as future work.

Prefix     probability     computation     is     mostly     studied     about
PCFGs~\cite{Jelinek91,Stolcke95,Nederhof11a}.   \citeN{Jelinek91}  proposed  a
CKY  like algorithm  for prefix  probability computation  in PCFGs  in Chomsky
normal  form.  Their algorithm  does not  perform parsing  but instead  uses a
single monolithic matrix  whose dimension is the number  of nonterminals which
is constructed from a given PCFG.  It runs in $O(N^3)$ where $N$ is the length
of an  input prefix.   \citeN{Stolcke95} applied the  Earley style  parsing to
compute prefix probabilities.  His  algorithm uses a matrix of ``probabilistic
reflexive,  transitive  left-corner relation''  computed  from  a given  PCFG,
independently of input sentences  similarly to \cite{Jelinek91}.  Our approach
differs from  them first in that it  is general and works  for arbitrary PRISM
programs and second in that it  constructs an explanation graph for each input
prefix and  probabilities are computed on  the basis of the  SCCs derived from
the explanation graph.

\citeN{Nederhof11a}  generalized prefix probability  computation for  PCFGs to
infix probability computation for PCFGs.  They also studied prefix probability
computation   for  a   variant  of   PCFGs~\cite{Nederhof11b}.    Nederhof  et
al.\  proposed prefix  probability computation  for stochastic  tree adjoining
grammars~\cite{Nederhof98}. However, prefix  probability computation for PLCGs
has been unknown and our  example in Section~\ref{sec:preplc} is the first one
to our knowledge.

Applying   prefix    probability   computation   to    plan   recognition   in
Section~\ref{sec:plan_recog} is not new  but our approach generalizes previous
grammar-based
approaches~\cite{Bobick98,Lymberopoulos07,Amft07,Geib11,Pomponio11} in that it
allows  for  incomplete  action  sequences  (prefixes)  as  observations.   In
relation  to plan  recognition, it  is  possible to  apply prefix  probability
computation to predict the most  likely action (word) that follows an observed
action sequences (prefix)~\cite{Jelinek91}, though we do not discuss it here.

We eliminated in  this paper one of the restrictive  assumptions on PRISM that
the number  of explanations for a  goal is finite. However  there still remain
restrictive  assumptions, the  exclusiveness assumption  and  the independence
assumption~\cite{Sato01g}.       Their     elimination      by     introducing
BDDs~\cite{DeRaedt07,Riguzzi11} or FEDs~\cite{Gorlin12} remains future work.

\section{Conclusion}
\label{sec:conclusion}
We  have   proposed  an  innovative  use  of   tabling:  infinite  probability
computation based on  cyclic explanation graphs generated by  tabled search in
PRISM.  It  generalizes  prefix  probability  computation  for  PCFGs  and  is
applicable to probabilistic models described  by PRISM programs in general and
to  non-PCFG  probabilistic  grammars  such  as  PLCGs  in  particular  as  we
demonstrated.   We  applied  our  approach  to plan  recognition  and  to  the
reachability probability  problem in probabilistic model  checking.  We expect
that  our approach  provides a  declarative way  of  logic-based probabilistic
modeling of cyclic relations.

\end{document}